\documentclass[10pt,journal,compsoc]{IEEEtran}

\usepackage[dvipsnames,table]{xcolor}
\usepackage{colortbl}

\usepackage{amsmath,amssymb,amsfonts,mathtools}
\usepackage{amsthm}

\usepackage{graphicx}
\usepackage{booktabs}
\usepackage{multirow}
\usepackage{multicol}
\usepackage{makecell}
\usepackage{array}
\usepackage{tabularx}
\usepackage{arydshln}
\usepackage{adjustbox}
\usepackage{textcomp}

\usepackage[linesnumbered,ruled,vlined]{algorithm2e}
\SetKwInput{KwInit}{Init}

\usepackage{float}
\usepackage{placeins}
\usepackage{stfloats}

\usepackage{subfig}

\usepackage[most]{tcolorbox}
\usepackage{mdframed}

\usepackage{url}
\usepackage{cite}
\usepackage{listings}
\usepackage{verbatim}

\usepackage[colorlinks=true,
    linkcolor=MidnightBlue,
    citecolor=ForestGreen,
    urlcolor=BrickRed
]{hyperref}

\newtheorem{theorem}{Theorem}

\theoremstyle{plain}

\definecolor{titlegray}{HTML}{555555}
\definecolor{framegray}{gray}{0.6}
\definecolor{lightgray}{gray}{0.97}
\newtcolorbox{promptbox}[1][]{
    enhanced,
    colback=lightgray,
    colframe=framegray,
    coltitle=white,
    colbacktitle=titlegray,
    fonttitle=\bfseries\itshape\normalsize, 
    title=#1,
    attach boxed title to top left={xshift=0.6cm,yshift=-3mm},
    boxed title style={
        boxrule=0.8pt,
        colframe=titlegray,
        top=3pt,
        bottom=3pt,
        left=8pt,
        right=8pt,
        rounded corners=south,
    },
    width=0.45\textwidth,
    before skip=6pt,
    after skip=12pt,
    boxrule=0.8pt,
    sharp corners=south,
    rounded corners,
    top=8pt,
    bottom=6pt,
    left=8pt,
    right=8pt,
    fontupper=\normalsize,
    coltext=black,
}

\lstdefinelanguage{diff}{
  morecomment=[f][\color{gray}]{@@},
  morecomment=[f][\color{gray}]{---},
  morecomment=[f][\color{gray}]{+++},
  morecomment=[f][\color{red}]{-},
  morecomment=[f][\color{green!60!black}]{+},
}

\lstdefinestyle{diff}{
  language=diff,
  basicstyle=\ttfamily\footnotesize,
  columns=fullflexible,
  keepspaces=true,
  showstringspaces=false,
  breaklines=true
}

\begin{document}
\renewcommand{\thesection}{\arabic{section}}
\renewcommand{\thesubsection}{\thesection.\arabic{subsection}}
\renewcommand{\thesubsubsection}{\thesubsection.\arabic{subsubsection}}

\title{From Illusion to Insight: Change-Aware File-Level Software Defect Prediction Using Agentic AI}  

\author{
    \IEEEauthorblockN{Mohsen Hesamolhokama$^1$\IEEEauthorrefmark{2,1}\IEEEauthorrefmark{1,1}, Behnam Rohani$^2$\IEEEauthorrefmark{1,1}, Amirahmad Shafiee$^2$, MohammadAmin Fazli$^1$, Jafar Habibi$^1$}\\
    \IEEEauthorblockA{$^1$Department of Computer Engineering, Sharif University of Technology, Tehran, Iran} \\
    \IEEEauthorblockA{$^2$Department of Mathematical Sciences, Sharif University of Technology, Tehran, Iran}\\
    {Emails: hokama@ce.sharif.edu, \{behnam.rohani058 , amirahmad.shafiee, fazli, jhabibi\}@sharif.edu}
    \thanks{*These authors contributed equally to this work.}
    \thanks{\dag Corresponding author (email: hokama@ce.sharif.edu).}
}
\maketitle

\begin{abstract}
Much of the reported progress in file-level software defect prediction (SDP) is, in reality, nothing but an \textit{illusion of accuracy}.
Over the last decades, machine learning and deep learning models have reported increasing performance across software versions.
However, since most files persist across releases and retain their 
defect labels, standard evaluation rewards \textit{label-persistence bias} rather than reasoning about code changes. To address this issue, we reformulate SDP as a \textit{change-aware prediction task}, in which models reason over code changes of a file within successive project versions, rather than relying on static file snapshots. Building on this formulation, we propose an LLM-driven, change-aware, multi-agent debate framework. Our experiments on multiple PROMISE projects show that traditional models achieve inflated F1, while failing on rare but critical defect-transition cases. In contrast, our change-aware reasoning and multi-agent debate framework yields more balanced performance across evolution subsets and significantly improves sensitivity to defect introductions. 
These results highlight fundamental flaws in current SDP evaluation practices and emphasize the need for change-aware reasoning in practical defect prediction. 
The source code is publicly available.\footnote{\url{https://github.com/mhhokama/from-illusion-to-insight}}

\end{abstract}
\maketitle

\section{Introduction}
\label{sec:sample1}
As software systems become central to essential services in modern society, their reliability is more important than ever. These systems support major infrastructures such as healthcare, finance, transportation, and communication. However, the growing complexity of software and the constant evolution of codebases introduce defects that may cause serious failures, financial losses, or safety risks. Studies estimate that software quality assurance can account for up to 50\% of total development cost, highlighting the need for early and accurate defect detection \cite{Li2018, Hall2012, Kitchenham2017}. As a result, the research community has invested significant effort in Software Defect Prediction (SDP), which aims to identify defect-prone modules before deployment to reduce maintenance costs and post-release failures.

Traditional SDP approaches mostly rely on supervised models trained on historical labeled data. These include logistic regression, decision trees, and support vector machines \cite{Rahman2013, Tantithamthavorn2018, Hosseini2017}. Although their performance is mistakenly reported as effective in prior works, the predictions often create an \textit{illusion of accuracy}. This occurs because overlapping files inflate performance metrics due to label-persistence bias. Consequently, models memorize historical labels instead of reasoning about code changes.

Despite the availability of widely used benchmark datasets such as PROMISE\cite{promise_awsm}, NASA\cite{nasa_klainfo} MDP, and AEEEM\cite{aeeem_kaggle}, current SDP models still encounter fundamental challenges when applied to evolving software projects\cite{nevendra2022survey}. 
These datasets were constructed for static cross-version validation, in which models are trained on one version and tested on another within the same project. As source code is textual data, it requires numerical features for model training. Recent advances in deep learning have enabled models to learn semantic representations directly from raw code and commit data. CNN, RNN, and Transformer-based architectures outperform metric-driven models \cite{zhou2024research, liu2025improved, Majd2021}. Models such as ASTNN \cite{Zhou2019} and DeepJIT \cite{Hoang2020} capture program syntax and commit-level semantics. Graph-based methods \cite{Dam2019, Adediran2024} also incorporate structural dependencies to improve defect reasoning. Pretrained code-focused embeddings, including CodeBERT \cite{Feng2020}, CodeT5+ \cite{Wang2023}, XLNet \cite{yang2019xlnet}, and StarCoder2 \cite{lozhkov2024starcoder}, learn jointly from source code and natural language. 

However, these conventional approaches, along with their evaluation, are inherently flawed since they don't take into account label-persistence bias in overlapping files.
Modern LLMs somewhat address this issue through in-context reasoning \cite{wei2022chain} since they don't necessarily need to be fine-tuned on previous historical data. They synthesize code semantics and descriptions to estimate defect likelihoods \cite{nam2024using}. Despite this progress, most deep learning and LLM-based methods still treat software versions as static snapshots. They rely on metrics or fixed embeddings that ignore the temporal and causal dynamics of code evolution. Consequently, they struggle to distinguish harmless refactorings from defect-inducing modifications \cite{clinton2025proactive, deeptha2024robust}. Traditional one-step classification pipelines further limit interpretability because they predict labels without explaining \textit{why} a change may introduce risk. 
These challenges point to the need for a \emph{change-aware and evolution-sensitive} approach, which avoids the aforementioned \textit{illusion} by grounding prediction in the differences introduced between versions~\cite{fakih2025llm4cve, yu2025graphrag, zeng2021deep}. 

Our proposed approach transforms software defect prediction from a single-file classification task into a change-aware task. The objective is to predict changes in defect status based on modifications in the source code, requiring knowledge of each file’s previous status. Unlike traditional methods, our approach emphasizes how code changes introduce or resolve defects across versions. Inspired by recent advances in multi-agent debate frameworks for problem-solving and software issue resolution \cite{li2025swe, tillmann2025literature},~\cite{vamosi2025crawdad} we extend beyond a single LLM by employing a competitive multi-agent debate structure. 
Our framework consists of four agents, each with a distinct role. At the beginning, an \emph{Analyzer} examines the code changes using \emph{dual-sided reasoning}. It 
simultaneously considers why a file may become defective and why it may remain or become benign after the change. Based on this analysis, a \emph{Proposer} formulates hypotheses about possible defect introductions or resolutions. These hypotheses are then challenged by an opposing \emph{Skeptic}. At the end, a \emph{Judge} assesses the competing arguments and produces the final prediction. As demonstrated in our experiments, this structured multi-agent reasoning improves performance on rare but critical defect-transition cases that are systematically overlooked by traditional and single-agent SDP approaches.

\subsection{Motivation Example}
In software engineering, a \textit{defect} represents a mismatch between human intention and program behavior during execution. Consider a simple case where a function aims to compute the sum of all elements in an array. Although the implementation appears correct, an off-by-one condition in the loop causes an \texttt{ArrayIndexOutOfBoundsException} at runtime. After correcting the loop boundary, the function performs as intended. This example illustrates that a defect is not merely a syntactic mistake. Understanding and predicting such issues requires reasoning about program logic and execution.

\lstdefinestyle{javaColored}{
  language=Java,
  basicstyle=\ttfamily,  
  keywordstyle=\color{blue}\bfseries,
  stringstyle=\color{red},
  commentstyle=\color{red}\itshape,
  numbers=left,
  numberstyle=\tiny\color{gray},
  stepnumber=1,
  numbersep=5pt,
  tabsize=4,
  showspaces=false,
  showstringspaces=false,
  breaklines=true,
  frame=none,
  backgroundcolor=\color{white}
}
\vspace{10pt} 

\setcounter{figure}{0}
\begin{figure*}[ht]

    \begin{minipage}[t]{0.48\textwidth}

     \label{fig:defect}   
        
        \begin{adjustbox}{width=\textwidth}
            \begin{lstlisting}[style=javaColored]
public class DefectExample {
    public static int calculateSum(int[] arr) {
        int total = 0;
        // Defective loop condition
        for (int i = 0; i <= arr.length; i++) {
            total += arr[i];
        }
        return total;
    }

    public static void main(String[] args) {
        int[] numbers = {1, 2, 3, 4, 5};
        System.out.println(calculateSum(numbers));
    }
}
            \end{lstlisting}
        \end{adjustbox}
        \centering
        \textbf{(a) Defective Code}
    \end{minipage}
    \hfill
    \begin{minipage}[t]{0.48\textwidth}
        
       \label{fig:fixed} 
        \begin{adjustbox}{width=\textwidth}
            \begin{lstlisting}[style=javaColored]
public class FixedExample {
    public static int calculateSum(int[] arr) {
        int total = 0;
        // Corrected loop condition
        for (int i = 0; i < arr.length; i++) {
            total += arr[i];
        }
        return total;
    }

    public static void main(String[] args) {
        int[] numbers = {1, 2, 3, 4, 5};
        System.out.println(calculateSum(numbers));
    }
}
            \end{lstlisting}
        \end{adjustbox}
        \centering
        \textbf{(b) Corrected Code}
    \end{minipage}
    \centering
    \caption{Comparison of Defective and Corrected Java Code}
    \label{fig:compare}
\end{figure*}

\subsection{Contributions}
The main contributions of this paper are as follows: 

\begin{itemize}
    \item \textbf{Analysis of "Illusion of Accuracy" in SDP.}
    We expose an \emph{illusion of accuracy} in Tranditional SDP, showing that reported performance gains largely come from file overlap, label persistence, and \emph{wrong in being right} phenomenon.
    \item \textbf{A Change-aware Formulation and Evaluation in SDP.}
    We introduce a change-aware formulation of SDP that models file evolution through status-transition subsets. 
    This is used to evaluate LLMs under multiple change-aware reasoning methods.
    \item \textbf{A Multi-agent Debate Framework for Change-aware Defect Prediction.}
    We propose a change-aware, multi-agent debate framework that reasons directly over code changes. 
    This framework moves beyond the deceptive performance of traditional models, and takes the first steps toward a balanced performance on the four status-transition subsets. 
\end{itemize}

\subsection{Paper Outline}
The remainder of the paper is organized as follows. Section~\ref{sec:related_works} describes background and related work. Section~\ref{sec:methodology} presents our change-aware formulation, file-matching strategy, context-expansion algorithm, and multi-agent debate framework. Section \ref{sec:experimental_design} details the dataset, baselines, prompting strategies, and evaluation setup. Section \ref{sec:results} reports the findings and answers the research questions. Section~\ref{sec:discussion} provides theoretical justification of change-aware perspective, along with an in-depth comparative analysis of our proposed methodology. Section \ref{sec:validity} describes threats to validity. Section \ref{sec:conclusion} concludes the paper.

\section{‌Background And Related Work}
\label{sec:related_works}
\subsection{Software Defect Prediction}
Software Defect Prediction (SDP) has long been a central research topic in software engineering, aiming to identify defect-prone modules before failures lead to increased maintenance costs and reduced reliability. Early SDP studies predominantly relied on supervised machine learning algorithms such as Decision Trees, Support Vector Machines, and Logistic Regression \cite{lessmann2008benchmarking, hall2012systematic}. These approaches utilized static software metrics, including cyclomatic complexity, coupling, lines of code, and historical change information \cite{gray2010software}. However, these foundational, traditional models faced notable limitations; they struggled with high-dimensional feature spaces \cite{elish2008predicting}, were highly sensitive to severe class imbalance \cite{bennin2018mahakil}, and largely ignored deeper semantic properties of code, such as logical dependencies and control-flow relationships \cite{liu2024cfg2at}. Consequently, these models often overfit project-specific characteristics and exhibited limited generalization to unseen or evolving software systems.

As the limitations of traditional machine learning became evident, deep learning emerged as a transformative approach for SDP by enabling automatic feature learning directly from raw code artifacts. Deep architectures such as Convolutional Neural Networks (CNNs)~\cite{o2015introduction}, Recurrent Neural Networks (RNNs)~\cite{sherstinsky2020fundamentals}, and Transformer-based models can capture complex syntactic and semantic patterns without relying on handcrafted features. Wang et al. \cite{wang2018deep} introduced a deep semantic feature learning framework that integrates structural and contextual information, significantly outperforming conventional classifiers. Majd et al. \cite{majd2020sldeep} proposed SLDeep to model statement-level semantics, while Pornprasit and Tantithamthavorn \cite{pornprasit2023deeplinedp} extended deep learning to line-level defect prediction. While these approaches improved semantic awareness, they still depend heavily on large volumes of labeled training data, limiting their effectiveness in projects with scarce or incomplete defect annotations.

To further enhance predictive capability and interpretability, recent SDP research has increasingly emphasized semantic feature learning, recognizing that software quality depends not only on structural complexity but also on developer intent and program logic. Earlier syntactic metric–based models failed to adequately represent logical flow and semantic context. To address this gap, researchers have incorporated contextual information from comments, commits, and issue reports alongside source code. Liu et al. \cite{liu2023semantic} proposed a unified framework that jointly models source code and natural language artifacts to better capture developer intent. Xu et al. \cite{xu2020defect} demonstrated that graph neural networks operating on abstract syntax trees (ASTs) effectively encode both structural and semantic dependencies, while Huo et al. \cite{huo2018learning} showed that comment-based embeddings significantly enhance defect prediction accuracy. Collectively, these studies demonstrate that modeling the semantic interplay between code structure and developer reasoning leads to more robust, interpretable, and context-aware SDP models.

Despite these advances, many real-world projects suffer from insufficient labeled data, motivating extensive research on Cross-Project Defect Prediction (CPDP). CPDP aims to transfer knowledge from well-labeled source projects to target projects with limited or no defect labels~\cite{saeed2024cross}. Kamei et al. \cite{kamei2016just} demonstrated that cross-project JIT defect prediction is feasible when sufficient domain similarity exists. Bai et al. \cite{bai2022transfer} further proposed a multi-source transfer learning framework to improve generalization across heterogeneous projects. Nevertheless, CPDP remains challenging due to variations in coding standards, development practices, and project scale. To mitigate these issues, researchers have explored adaptive feature weighting \cite{tong2023array}, adversarial domain adaptation \cite{sheng2020adversarial}, and graph-based transfer learning \cite{xing2022cross}. However, effectively aligning feature distributions across disparate projects continues to be a major bottleneck.

In addition to data scarcity, class imbalance represents a fundamental challenge in SDP, as defective modules are vastly outnumbered by non-defective ones. This imbalance biases learning algorithms toward the majority class, often resulting in low defect recall~\cite{gong2019tackling}. To address this issue, both data-level and algorithm-level solutions have been proposed. Oversampling techniques such as SMOTE \cite{pak2018smote} generate synthetic minority instances, while Feng et al. \cite{feng2021coste} introduced Coste, a complexity-based oversampling method that prioritizes hard-to-classify modules. Cost-sensitive deep learning approaches further mitigate imbalance by penalizing false negatives more heavily. Yedida and Menzies \cite{yedida2021value} demonstrated that combining rebalancing strategies with ensemble deep models significantly improves robustness. However, oversampling methods still risk overfitting and artificial diversity, underscoring the need for more adaptive and semantically informed rebalancing techniques.

To improve robustness and stability, hybrid architectures and ensemble learning have been widely adopted in SDP. Tong et al. \cite{tong2018software} integrated stacked denoising autoencoders within an ensemble framework to enhance prediction consistency across software releases. Zhao et al. \cite{zhao2021metaheuristic} applied metaheuristic optimization to enable adaptive feature selection in deep neural networks, while Tong et al. \cite{tong2024master} proposed a multi-source weighted ensemble framework for cross-project prediction. In parallel, attention-based models have gained traction due to their ability to improve interpretability and focus learning on defect-relevant code regions. Yu et al. \cite{yu2021hierarchical} introduced DP-HNN, a hierarchical attention model operating over AST subtrees, and Zheng et al. \cite{zheng2021software} employed Transformer architectures to capture both local and global dependencies. These studies reflect a broader shift toward multi-perspective learning that integrates structural, semantic, and contextual information.

Building upon these directions, recent studies have further expanded deep learning for SDP through diverse architectures and learning paradigms. Pandey et al.~\cite{pandey2025deep} evaluated SqueezeNet and Bottleneck models on NASA datasets, achieving F-measures of up to 0.93, albeit with high computational costs. Hu et al.~\cite{hu2025fed} proposed Fed-OLF, a federated oversampling framework that improves F1, G-mean, and AUC while preserving data privacy. Han et al.~\cite{han2024bjcnet} introduced bjCnet, which leverages supervised contrastive learning with Transformer-based code models, achieving 0.948 accuracy and F1. Hesamolhokama et al.~\cite{hesamolhokama2024sdperl} combined ensemble feature extraction from pre-trained code models with reinforcement learning–based feature selection. Malhotra et al.~\cite{malhotra2025dhg} proposed DHG-BiGRU with dual attention over static and AST-based semantic features, while Nashaat et al.~\cite{nashaat2025refining} applied bidirectional Transformers with attention. Overall, these studies highlight a clear trend toward hybrid, attention-driven, and adaptive SDP models that jointly exploit semantic, structural, and contextual information.

\subsection{Source Code Representation Learning}
Effective software defect prediction depends critically on how source code is represented for learning models. Traditional metric-based features capture superficial code statistics but fail to model syntactic and semantic dependencies. Early representation learning approaches, such as \textit{code2vec} and \textit{code2seq}, introduced path-based embeddings from Abstract Syntax Trees (ASTs), enabling structural awareness and partial semantic understanding of source code \cite{alon2019code2vec, alon2018code2seq}. Building on this idea, \textit{ASTNN} \cite{zhang2019astnn} segmented large ASTs into statement-level subtrees to better capture local code patterns.
\par
With the rise of pre-trained language models, \textit{CodeBERT} \cite{feng2020codebert} and \textit{GraphCodeBERT} \cite{guo2021graphcodebert} revolutionized code representation by jointly encoding code tokens and data-flow dependencies. These models substantially improved downstream tasks such as defect prediction, code summarization, and vulnerability detection by incorporating both structural and contextual semantics. Recent works have continued this trend: \textit{TransformCode} \cite{xian2024transformcode} leveraged contrastive learning over transformed subtrees to enhance embedding robustness, while \textit{xASTNN} \cite{xu2023xastnn} refined AST segmentation to improve representation generalization in industrial settings.

More recent studies have emphasized hybrid representations that combine textual and graph-based information to capture the ``mixed polysemy'' nature of source code \cite{li2023contextuality, wang2023tree}. However, these models remain primarily static, focusing on fixed code snapshots without considering evolving files or semantic changes over time. This limitation is particularly critical for file-level SDP, where changes between software versions fundamentally alter defect behavior. Thus, while deep representation learning has advanced the expressiveness of code embeddings, there remains a need for \textbf{change-aware, semantics-grounded representations} that can reason about modified code and its defect potential in dynamic software environments.

\subsection{LLMs for Code}
The emergence of Large Language Models (LLMs) such as CodeBERT, CodeT5, and GPT-3/4 has opened new possibilities for software engineering research, particularly in source code comprehension, bug detection, and reasoning about software behavior. Unlike traditional supervised approaches that rely on static code metrics or historical defect labels, LLMs can analyze source code by capturing semantic dependencies and natural-language-like structure \cite{feng2020codebert, ahmad2021unified}.

Recent studies demonstrate that LLMs are capable of identifying vulnerabilities, generating bug fixes, and understanding complex control flows through pre-trained knowledge of programming languages and natural language \cite{chen2021evaluating, hu2023large}. However, leveraging these capabilities effectively for SDP remains underexplored. While most prior work has focused on code completion or repair tasks, the use of prompt engineering enables LLMs to reason about code quality without requiring labeled data. Prompting techniques allow researchers to frame the prediction process as a question–answer task, guiding the model to behave like an expert reviewer who examines whether a piece of code is likely defective \cite{white2023chatgpt, peng2023impact}.

\section{Methodology}
\label{sec:methodology}

We present the structure of our multi-agent debate framework with context expansion algorithm in the following subsections.  

\subsection{Motivation}
\label{sec:motivation}
Consecutive versions of a codebase often share many files~\cite{thota2020survey}, leading to an illusion of accuracy in traditional defect prediction models that treat files independently. 
Even with Large Language Models (LLMs), reasoning over changes alone has limitations, as many defects arise from how modified code sections interact with the rest of the file.
To address this, it is necessary to adopt a change-aware formulation that reasons about how code modifications affect program behavior. 

\subsection{Problem Formulation}
\label{sec:problem_formulation}

We study the problem of predicting whether a file becomes defective in the next version of a software project. Let a project evolve through versions
\(V_1, V_2, \dots, V_T\), where
\(V_t = \{ x^{(t)}_i \}_{i=1}^{N_t}\) and each file $x^{(t)}_i$ has a defect label
\(y^{(t)}_i \in \{0,1\}\), with 0 representing benign code and 1 representing defective code.

\subsubsection{Traditional formulation}

Classical SDP methods train a model on an
earlier version \(V_t\) and evaluate it on the subsequent version
\(V_{t+1}\). The model parameters \(\theta\) are typically learned by
minimizing a standard loss
\begin{align}
\theta^\star = \arg\min_\theta 
\sum_{i} \ell(f_\theta(x^{(t)}_i),\, y^{(t)}_i),
\end{align}
where the goal is to learn a function \(f_\theta\) that maps a file’s
source code \(x\) to its defect label \(y\)~\cite{akimova2021survey}. After training, the predictor is applied to the test set
$\{(x^{(t+1)}_i, y^{(t+1)}_i)\}_{i=1}^{N}$. This formulation treats files in
\(V_{t+1}\) independently, without considering how those files evolved
from \(V_t\). 

\subsubsection{File-evolution partition}
Between \(V_t\) and \(V_{t+1}\), files fall into three categories:

\paragraph*{Removed files}
\begin{align}
\mathcal{R}_{t,t+1}
= \{\, i \mid x^{(t)}_i \notin V_{t+1} \,\},
\end{align}
typically rare and excluded from prediction.

\paragraph*{Added files}
\begin{align}
\mathcal{A}_{t,t+1}
= \{\, j \mid x^{(t+1)}_j \notin V_t \,\},
\end{align}
also excluded in our work.

\paragraph*{Common files}
\begin{align}
\mathcal{C}_{t,t+1}
= \{\, i \mid x^{(t)}_i \text{ exists in } V_{t+1} \,\},
\end{align}
which constitute the majority (often \(>80\%\)) of all files.  

For each common file \(i\), define
$
d^{(i)} = \operatorname{diff}(x^{(t)}_i, x^{(t+1)}_i)
$.
Files with unchanged source code (\(d^{(i)} = \emptyset\)) trivially preserve their labels and are uninformative. We therefore focus on changed-source files:
\begin{align}
\mathcal{C}^{\text{chg}}_{t,t+1}
= \{\, i \in \mathcal{C}_{t,t+1} \mid d^{(i)} \neq \emptyset \,\}.
\end{align}
\subsubsection{Evolution subsets for changed-source files}
Among files whose source code changes between versions, each file falls
into exactly one of four status–transition subsets:
\[
\begin{aligned}
\mathrm{Bj0} &= \{\, i \mid y^{(t)}_i=j,\; y^{(t+1)}_i=0 \,\}\quad j \in \{0,1\},\\
\mathrm{Dj1} &= \{\, i \mid y^{(t)}_i=j,\; y^{(t+1)}_i=1 \,\}\quad j \in \{0,1\}.\\
\end{aligned}
\]
For evaluation, we merge them into two higher-level categories:
\[
\begin{aligned}
\mathcal{S}_{\text{unchanged}} &= \mathrm{B00} \cup \mathrm{D11} \quad ,\quad\mathcal{S}_{\text{changed}}   &= \mathrm{B10} \cup \mathrm{D01}.
\end{aligned}
\]
In practice, the size of the status-changed subset is extremely small compared to the total number of common files:
\begin{align}
|\mathcal{S}_{\text{changed}}| \ll |\mathcal{C}_{t,t+1}|.
\end{align}

\subsubsection{Prediction task}
For each changed-source file \(i\), a general model receives 
\[
(x^{(t)}_i,\, x^{(t+1)}_i,\, y^{(t)}_i)
\]
and must predict
\begin{align}
\hat{y}^{(t+1)}_i = f(x^{(t)}_i, x^{(t+1)}_i, y^{(t)}_i).
\end{align}
To model code evolution, we introduce the change \(d^{(i)}\) between \(x^{(t)}_i\) and \(x^{(t+1)}_i\) and define a context expansion function 
\begin{align}
    z_i = g(d^{(i)}, x^{(t+1)}_i),
\end{align}
where \(g(\cdot)\) acts as an estimator of the most relevant information in \((x^{(t)}_i, x^{(t+1)}_i)\) by extracting from \(x^{(t+1)}_i\) the code regions semantically related to the modified lines. Under this change-focused formulation, the prediction is given by
\begin{align}
\hat{y}^{(t+1)}_i = f(d^{(i)}, z_i, y^{(t)}_i).
\end{align}
The task is therefore to determine whether the edits \(d^{(i)}\) introduce, remove, or preserve a defect.
\subsection{File Matching Across Versions}
\label{sec:file_matching}

To reason about file evolution, we must first determine which files in
\(V_{t+1}\) correspond to files in \(V_t\). Relying solely on file paths
is fragile, since files may be renamed or relocated without changing
their underlying content. We therefore establish correspondences using a
matching procedure that combines path agreement with similarity-based
alignment. The full matching rule is summarized in Algorithm~\ref{alg:file_matching}. 
For the similarity function, we use the normalized overlap
\[
S(x,y)
=
\frac{2\,|\mathrm{lines}(x)\cap\mathrm{lines}(y)|}
     {|\mathrm{lines}(x)| + |\mathrm{lines}(y)|}.
\]
Under mild assumptions, namely that (i) an evolved file retains a significant portion of its ancestor’s lines and (ii) unrelated files share only incidental overlaps, there exists a similarity threshold~$T$ and a gap confidence multiplier~$c$ such that, whenever a genuine predecessor exists, the algorithm identifies it with probability at least $1-\delta$, and when no predecessor exists, the probability of incorrectly declaring a match is at most~$\delta$, where $\delta$ can be made arbitrarily small by choosing appropriate values of $T$ and~$c$ (Check \ref{appendix:matching_proof} for proof).

\begin{algorithm}[t]
\caption{SimilarityBasedFileMatching}
\label{alg:file_matching}

\KwIn{Old-version files $D_{\text{old}}$, new-version files $D_{\text{new}}$, similarity function $\mathrm{sim}$, threshold $T$, constant $c \ge 1$}
\KwOut{A mapping from new-version files to matched old-version files}

matches $\gets$ empty map\;

\ForEach{$x_2 \in D_{\text{new}}$}{

    \If{path$(x_2)$ exists in $D_{\text{old}}$}{
        matches[$x_2$] $\gets$ old-version file with identical path\;
    }
    \Else{

        compute scores $s(z) = \mathrm{sim}(x_2, z)$ for all $z \in D_{\text{old}}$\;

        sort scores as $s_1 \ge s_2 \ge \cdots \ge s_m$\;

        compute gaps $g_i = s_i - s_{i+1}$ for $1 \le i < m$\;

        $\mu_g \gets \mathrm{mean}(g_i)$\;

        $\sigma_g \gets \mathrm{std}(g_i)$\;

        \If{$s_1 \ge T$ \textbf{and} $(s_1 - s_2) \ge \mu_g + c \cdot \sigma_g$}{
            matches[$x_2$] $\gets$ file corresponding to $s_1$\;
        }
        \Else{
            matches[$x_2$] $\gets$ \textbf{Null}\;
        }
    }
}

\Return{matches}

\end{algorithm}

\subsection{Data Preparation: Context Expansion Algorithm}
\label{sec:subset_motivation_retrieval}

We extend diff-based analysis with an approach that expands the available
context around a change. Starting from the unified diff (Appendix~\ref{appendix:unified_diff_exp}), we locate the methods
that contain the modified lines. We then include the methods that directly call them as well as the methods they call. This expansion provides the model with nearby behavioral information that is often relevant for understanding the change but does not appear in the diff itself. Algorithm~\ref{alg:context_pipeline} summarizes the full procedure. 

\begin{footnotesize}
\begin{algorithm}[t]
\caption{ExtractFileLocalContext(diffs, src, depth, max\_lines)}
\label{alg:context_pipeline}

\KwIn{Unified diff \textit{diffs}, source \textit{src}, expansion depth $\ge 0$, \textbf{max\_lines}}
\KwOut{Context snippet}

changed\_lines $\gets$ extract line numbers from \textit{diffs}$;$

methods $\gets$ find all methods in \textit{src}$;$

initial $\gets$ methods overlapping with \textit{changed\_lines}$;$

callgraph $\gets$ construct a call graph by identifying caller--callee relationships between methods in \textit{src}$;$

visited $\gets \emptyset$; queue $\gets$ initial$;$

\While{queue not empty \textbf{and} depth $\ge 0$}{
    frontier $\gets$ queue; clear queue$;$

    \ForEach{method $m$ in frontier}{
        add $m$ to visited$;$ 
        enqueue all callers and callees of $m$ in the callgraph$;$ 
    }
    depth $\gets$ depth $- 1$;
}

snippets $\gets$ build contextual code snippets for all methods in \textit{visited},
including method signatures and surrounding source lines$;$

output $\gets$ concatenate snippets$;$

truncate output to \textbf{max\_lines}$;$

\Return output$;$

\end{algorithm}
\end{footnotesize}

\subsection{Multi-agent Debate Framework}
\label{sec:debate_arch}

Our approach replaces single-pass LLM prediction with a structured debate among multiple LLM-based agents. Inspired by prior work on multi-agent debate and argumentation frameworks \cite{kontarinis2014debate,liang2024encouraging,liu2024groupdebate}, the framework is centered on an adversarial exchange in which agents reason from opposing perspectives and challenge one another’s conclusions. The outcomes of this adversarial process are then put together to produce a single resolved prediction, mediated by a \emph{Judge} agent. \autoref{fig:debate_arch} illustrates the interactive process and information flow among agents in this framework.

\begin{figure*}[t]
    \centering
    \includegraphics[width=0.80\linewidth]{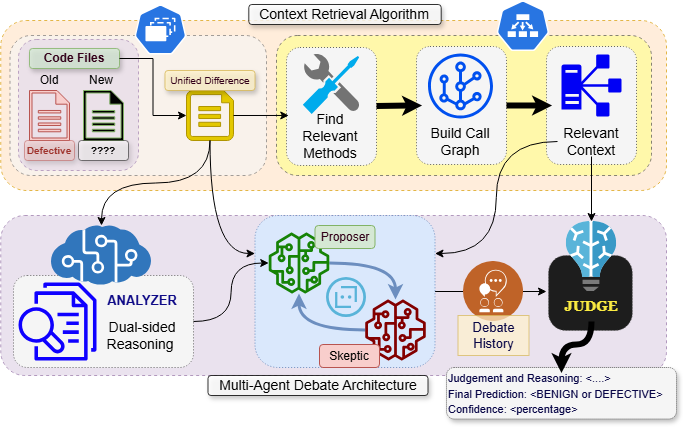}
    \caption{Overview of the multi-agent debate architecture.}
    \label{fig:debate_arch}
\end{figure*}

\subsubsection{Data Preparation and Debate Setup}

Each evaluation instance provides three elements:  
(1) a unified diff describing the code edits,  
(2) the expanded file-local context from Section~\ref{sec:subset_motivation_retrieval},  
and (3) the previous label.  
These form the shared evidence the agents use when debating the updated status of the file.

\subsubsection{Role Descriptions and Responsibilities}
The debate uses four LLM agents, each with a restricted prompt and access only to the information required for its role:

\begin{itemize}
    \item \textbf{Analyzer (Round 0).}  
    Receives the diff and expanded context and produces a summary of the changes and lists plausible benign and defective 
    interpretations without assigning a label.

    \item \textbf{Proposer (Rounds 1--$R$).}  
    Receives the diff, context, and the Analyzer’s output in the first round, plus possibly earlier debate messages.
    States whether the updated file should keep or change its previous label and supports the 
    claim with specific evidence.

    \item \textbf{Skeptic (Rounds 1--$R$).}  
    Receives the diff, context, and earlier debate messages.  
    Challenges the Proposer’s claim by pointing out weaknesses, alternatives, or missed cases.

    \item \textbf{Judge (Final Stage).}  
    Receives the diff, context, previous label, and the full debate history.  
    Weighs the arguments and produces the final prediction with a confidence score.
\end{itemize}

\subsubsection{Rationale and Effects}
This architecture frames prediction as an adversarial process rather than a direct inference step.  
The agents confront one another's claims and require that any 
assertion be supported by explicit reasoning.  
A \emph{benign} or \emph{defective} label prediction is accepted only after the system produces a 
coherent argument that survives criticism from the participating roles.  
The full exchange remains accessible after the decision, which provides a clear record of the 
supporting logic and allows readers to trace the path from initial evidence to final judgment.

\section{Experimental Design}
\label{sec:experimental_design}

\subsection{Research Questions}
\label{sec:research_questions}

The objective of this study is to investigate the limitations of traditional 
file-level Software Defect Prediction (SDP) models and to evaluate the effectiveness 
of the proposed \textit{Change-Aware LLM-Based Method} that integrates multi-agent 
reasoning. To achieve these goals, 
we formulate three main research questions as follows:

\begin{itemize}
    \item \textbf{RQ1.} \textit{How severe is the illusion of accuracy in SDP by traditional models?}
    \vspace{1em}
    \item \textbf{RQ2.} \textit{How much can the performance of LLM-based methods increase, introducing change-aware context into their prompts?}
    \vspace{1em}
    \item \textbf{RQ3.} \textit{Can multi-agent debate help to solve the problem of \textbf{illusion of accuracy}?}
\end{itemize}

Each research question is empirically addressed in Section~\ref{sec:results}.

\subsection{Dataset Description}

To ensure comparability and reproducibility with prior research in software defect prediction, we employ the PROMISE repository, which includes several well-established open-source Java projects.
The selected projects represent large-scale, real-world Java systems spanning various domains and architectures, including:

\begin{itemize}
    \item \textbf{Derby} – a lightweight relational database engine,
    \item \textbf{Camel} – an open-source integration framework for routing and mediation,
    \item \textbf{Groovy} – a dynamic language for the Java platform,
    \item \textbf{HBase} – a distributed, scalable NoSQL database,
    \item \textbf{Lucene} – a high-performance text search library,
    \item \textbf{JRuby} – a Ruby interpreter implemented in Java,
\end{itemize}

Each project contains multiple historical versions of Java source files, accompanied by manually curated defect labels indicating whether a file was \textit{Defective} or \textit{Benign} in a given release. A summary of the datasets across project versions is provided in~\autoref{tab:version_stats_final}. In addition to the defect label column (Bug), every dataset includes a unique file identifier (File) and corresponding source code stored in the SRC field.

\begin{table}[H]
\centering
\caption{Dataset version statistics.}
\label{tab:version_stats_final}
\small
\begin{tabular}{l p{2cm} r r r}
\toprule
Dataset & Version & \#Files & \%Benign & \%Defective \\
\midrule

\multirow{2}{*}{lucene}
  & 2.9.0 & 1368 & 91.59 & 8.41 \\
  & 3.0.0 & 1337 & 94.39 & 5.61 \\
\midrule

\multirow{2}{*}{groovy}
  & 1\_5\_7 & 757 & 97.89 & 2.11 \\
  & 1\_6\_BETA\_2 & 884 & 96.38 & 3.62 \\
\midrule

\multirow{2}{*}{derby}
  & 10.3.1.4 & 2206 & 84.18 & 15.82 \\
  & 10.5.1.1 & 2705 & 93.90 & 6.10 \\
\midrule

\multirow{2}{*}{camel}
  & 2.10.0 & 7914 & 98.24 & 1.76 \\
  & 2.11.0 & 8846 & 98.47 & 1.53 \\
\midrule

\multirow{2}{*}{hbase}
  & 0.95.0 & 1669 & 89.33 & 10.67 \\
  & 0.95.2 & 1834 & 91.11 & 8.89 \\
\midrule

\multirow{2}{*}{jruby}
  & 1.4.0 & 978 & 87.01 & 12.99 \\
  & 1.5.0 & 1131 & 97.79 & 2.21 \\
\bottomrule
\end{tabular}
\end{table}

\subsection{Traditional Baselines}
\label{sec:traditional_baselines}

Traditional classifiers such as \text{Random Forest}~\cite{breiman2001random}, \text{Logistic Regression}~\cite{hosmer2013applied}, and \text{Support Vector Machine (SVM)}~\cite{cortes1995support} require fixed-length numerical vector inputs. Since the \textit{SRC} field contains raw source code, each file must be embedded before classification. We use two embedding methods: \text{CodeBERT}~\cite{feng2020codebert} and \text{OpenAI’s text-embedding model (small)}~\cite{openai2024textembedding3}.

CodeBERT accepts sequences of up to 512 tokens \cite{feng2020codebert}, so we choose to divide the longer files into overlapping chunks of 512 tokens with a stride of 256. Each chunk is encoded, and we apply token-level pooling (\text{CLS}, mean, max, or min) \cite{reimers2019sentence} to obtain one vector per chunk. Files with multiple chunks are then aggregated using the same pooling operations, producing a single 768-dimensional embedding. Each configuration is defined by the combination of token and chunk pooling:
\[
(\text{token pooling}) \times (\text{chunk pooling})
\]
In contrast, the OpenAI embedding model processes up to 8192 tokens at once and produces a 1536-dimensional file embedding\cite{openai2024textembedding3}. No token-level pooling is required.

\begin{figure}[h!]
    \centering
    \includegraphics[width=1.0\linewidth]{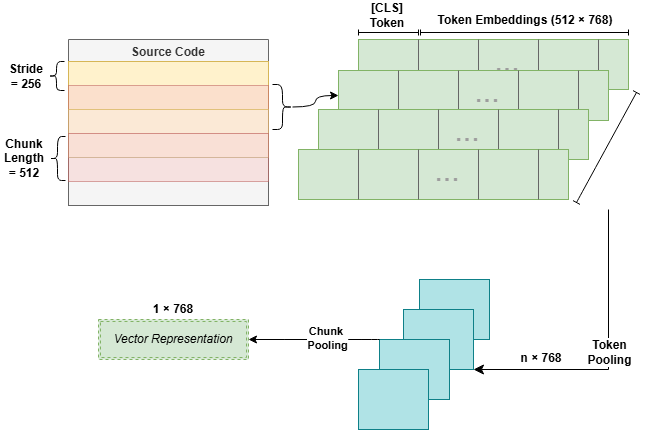}
    \caption{Overview of the embedding process for traditional baselines.}
    \label{fig:embedding_flow}
\end{figure}

The choice of training and testing versions is indicated in~\autoref{tab:version_stats_final}, following the standard temporal evaluation approach commonly used in software defect prediction \cite{menzies2007data, rahman2013how, yan2020software}. 

\subsection{Change-aware LLM Baselines}
\label{sec:llm_baselines_change}
We explore context-engineering\cite{brown2020language} by implementing nine LLM-based methods that differ in the type and amount of evolution-related information provided to the model. These methods range from single-version classification without historical information~\cite{li2024software} to full change-aware setting. 

\begin{enumerate}
\setcounter{enumi}{-1}

\item \textbf{Single-version classification:} 
The model receives only the target file (\text{SRC2}) and predicts whether it is Defective or Benign, as shown in~\autoref{fig:prompt_m0}.

\begin{figure}[ht]
\centering
\begin{promptbox}[Single-Version Classification]
You are given a source code file without any previous version. Decide if it is Defective or Benign.

[SRC2]
\textless source code here\textgreater

Question: What is the status of this file? (Defective or Benign)
\end{promptbox}
\caption{Prompt for Single-version classification.}
\label{fig:prompt_m0}
\end{figure}

\item \textbf{Direct comparison:} 
Both the previous version (\text{SRC1}, with known status) and the new version (\text{SRC2}) are provided.

\item \textbf{Diff-guided reasoning:} 
The model receives \text{SRC1}, \text{SRC2}, and a list of added, removed, and changed lines to focus on the modifications line by line.

\item \textbf{Patch-aware reasoning:} 
Along with both versions, the model is given a unified diff (patch) to reason about changes.

\item \textbf{SRC1-only change reasoning:} 
The model sees only \text{SRC1} and the diffs, without access to \text{SRC2}.

\item \textbf{Diff-only reasoning:} 
The model receives only the differences and the unified diff, along with the status of \text{SRC1}, to predict whether changes affect the defect state.

\begin{figure}[ht]
\centering
\begin{promptbox}[Diff-only Reasoning]
You are given only the differences and the unified diff. You also know that SRC1 was Defective. Determine if SRC2 remains Defective or changes status.

[SRC1] → Defective
[SRC2] → ???

[Differences]
\textless added/removed lines\textgreater

[Unified diff]
\textless patch-style diff\textgreater

Use this information to infer the status of SRC2.
\end{promptbox}
\caption{Prompt for Diff-only Reasoning.}
\label{fig:prompt_m5}
\end{figure}

\item \textbf{Local-context reasoning:} 
Only the modified lines and a few surrounding lines are provided.

\item \textbf{Diffs with exemplars:} 
The model receives the diff, unified diff, and a small set of known defective examples.

\item \textbf{Semantic repair reasoning:} 
A variation of Diff-only Reasoning.
\end{enumerate}

The full prompt templates for all LLM baselines are provided in Appendix~\ref{appendix:prompts_llm}.
In designing these prompts, we follow best practices in prompt engineering, such as providing clear and specific instructions \cite{wei2022chain}. 
Prior work has shown that clarity, context, ordering of input information, and step-by-step instructions can improve the consistency and quality of LLM outputs across complex tasks \cite{white2023prompt}. 
These design choices aim to help the model focus on the evolution-related information that is most relevant to defect status predictions \cite{liu2023pre}.

We evaluate a set of LLMs drawn from a variety of sources, including both open-source and proprietary models, as summarized in \autoref{tab:models}. All models are evaluated under a shared system prompt to ensure consistent behavior across baselines (Appendix~\ref{appendix:system_prompt}).

\begin{table}[!htbp]
\centering
\caption{LLMs evaluated as baselines for change-aware SDP. We additionally report \emph{parameters per token activated} (PPT), i.e., the number of parameters actively used for a single forward pass, shown in parentheses where applicable \cite{wan2023efficientllm}.}
\label{tab:models}
\renewcommand{\arraystretch}{1.3} 

\small
\setlength{\tabcolsep}{8pt}

\resizebox{\columnwidth}{!}{
\begin{tabular}{
p{2.6cm}@{\hspace{8pt}}
p{4.2cm}@{\hspace{8pt}}
p{2.3cm}
}
\hline
\textbf{Source} & \textbf{Model name} & \textbf{Params} \\
\hline

\multirow{3}{*}{DeepSeek}
 & \rule{0pt}{2.6ex}DeepSeek-Reasoning \cite{deepseek2025r1}
 & \multirow{3}{*}{671B (37B PPT)} \\
 & DeepSeek-V3.1 \cite{deepseek2024v3} &  \\
 & deepseek-v3.1:671b-terminus \cite{deepseek2025v31terminus} &  \\

\hline

Google
 & \rule{0pt}{2.6ex}Gemini-2.5-Flash-Lite \cite{google2025gemini25flashlite} & N/A \\

\hline

\multirow{4}{*}{OpenAI}
 & \rule{0pt}{2.6ex}GPT-5-Chat & N/A \\
 & GPT-5-Mini & N/A \\
 & GPT-5-Nano-2025-08-07 & N/A \\
 & GPT-O4-Mini-2025-04-16 & N/A \\

\hline

\multirow{5}{*}{Mistral AI}
 & \rule{0pt}{2.6ex}Mistral-Small-2503 \cite{mist} & 24B \\
 & Mistral-Small-3.1-24B-Instruct-2503 \cite{mistralai2025mistralsmallinstruct2503} & 24B \\
 & Open-Mixtral-8x7B \cite{jiang2024mixtral} & 46.7B \\
 & codestral-2405 \cite{mistralai2024codestral} & $\sim$22B \\
 & codestral-2501 \cite{mistralai2024codestral} & $\sim$22B \\

\hline

\multirow{2}{*}{Alibaba (Qwen)}
 & \rule{0pt}{2.6ex}Qwen2.5-Coder-32B-Instruct \cite{qwen2025coder32binstruct} & 32B \\

\hline
\end{tabular}
}
\end{table}

\subsection{Experimental Setting}
\label{sec:experimental_setting}
A complete summary of the resources, libraries, models, and parameter configurations used in our experiments is provided in Appendix~\ref{appendix:parameters_config}.

\subsection{Evaluation Setting}
\label{sec:res}
We conduct all sub-experiments for configuration tuning and model selection on the Lucene project (\autoref{tab:version_stats_final}), without loss of generality, as the insights we aim to establish are not project-dependent.
\subsubsection{LLM Output Parsing}

For LLM baselines, we enforce an output format to ensure reliable parsing. Each model is instructed to return its response as a JSON object:
{\small
\begin{verbatim}
{
  "explanation": "Explanation in English",
  "prediction": "Defective" or "Benign"
}
\end{verbatim}
}

In Multi-agent Debate, the final agent (Judge) is required to conclude its response with a prediction label and a confidence score.

{\small
\begin{verbatim}
### Final Prediction: <BENIGN or DEFECTIVE>
### Confidence: <confidence_percentage>
\end{verbatim}
}

LLM outputs may deviate from the prescribed format. When parsing fails, we apply regular-expression–based rules, refined through trial and error, to extract the final prediction.

\subsubsection{Standard Metrics}
Precision, Recall, and F1-score are computed using macro averaging; the term \emph{macro} is omitted hereafter for brevity. AUC is computed as the area under the receiver operating characteristic curve.

\subsubsection{Change-Aware Metrics}
For \textbf{RQ2} and \textbf{RQ3}, all metrics are computed over \emph{changed-source} files, i.e., common files with non-empty edits (\(d^{(i)} \neq \emptyset\)). We report:
\begin{align}
\text{unchanged-F1} &= \mathrm{F1}(\mathcal{S}_{\text{unchanged}}), \\
\text{changed-F1}   &= \mathrm{F1}(\mathcal{S}_{\text{changed}}),
\end{align}
along with the accuracy for each of the four evolution subsets (\(\mathrm{B00}\), \(\mathrm{B10}\), \(\mathrm{D01}\), \(\mathrm{D11}\)), as defined in Section~\ref{sec:problem_formulation}. 

We further define two harmonic-mean–based metrics~\cite{christen2023review, sokolova2009systematic, tharwat2021classification}, 
over the evolution subsets, where \(A_{\mathcal{S}}\) denotes the accuracy computed on subset \(\mathcal{S}\):
\begin{align}
\text{HMD} &=
\frac{2 \cdot A_{\mathrm{B10}} \cdot A_{\mathrm{D11}}}
{A_{\mathrm{B10}} + A_{\mathrm{D11}}}, \label{eq:hmb} \\
\text{HMB} &=
\frac{2 \cdot A_{\mathrm{B00}} \cdot A_{\mathrm{D01}}}
{A_{\mathrm{B00}} + A_{\mathrm{D01}}}. \label{eq:hmd}
\end{align}
Here, \textbf{HMB} (\emph{Harmonic Mean of Benign Prior Status}) captures balanced performance on files that were previously benign (\autoref{eq:hmb}), while \textbf{HMD} (\emph{Harmonic Mean of Defective Prior Status}) reflects balanced sensitivity on files with a defective history (\autoref{eq:hmd}).

\section{Experimental Results}
\label{sec:results}
\subsection{RQ1: Illusion Of Accuracy}
\label{sec:RQ1}

We report the file-evolution partitions for the PROMISE datasets in \autoref{tab:file_evolution_partition}. Across all projects, more than $78-85\%$ of the files are shared between versions. This is not surprising since software systems evolve incrementally, which naturally results in substantial file overlap between adjacent releases. This fact fundamentally undermines the traditional formulation of software defect prediction (SDP). Prior work on SDP has favored metrics such as F1-score, Matthews Correlation Coefficient (MCC), Precision, Recall, or False Positive Rate (FPR)~\cite{yang2025dlap}, since \textit{benign} class dominates defect status (see \autoref{tab:version_stats_final}).
While these metrics address the imbalance in defect labels, they overlook the imbalance that arises from file evolution across software versions. Consider the naive classification strategy in \autoref{fig:naive_classifier}. This naive classifier can achieve deceptively high values even under commonly used robustness-oriented metrics (\autoref{tab:naive_baseline}).
\begin{table}[!htbp]
\centering
\caption{File--evolution partition statistics (percentages) for consecutive
PROMISE versions in \autoref{tab:version_stats_final}. 
The terms $\mathcal{R}$ and $\mathcal{A}$ denote removed and added files, respectively.
$\mathcal{C}$ denotes common files, which are further divided into unchanged
($d=\emptyset$) and changed-source ($d\neq\emptyset$) files.}
\label{tab:file_evolution_partition}
\small
\setlength{\tabcolsep}{4pt}
\renewcommand{\arraystretch}{1.15}
\begin{tabular}{lcccccccc}
\toprule
\multirow{3}{*}{\textbf{Dataset}}
& \multirow{3}{*}{$\mathcal{R}$}
& \multirow{3}{*}{$\mathcal{A}$}
& \multicolumn{6}{c}{$\mathcal{C}$} \\
\cmidrule(lr){4-9}
& & 
& \multirow{2}{*}{$d=\emptyset$}
& \multicolumn{4}{c}{$d \neq \emptyset$} \\
\cmidrule(lr){5-9}
& & & 
& B00 & B10 & D11 & D01 \\
\midrule
lucene & 3.89 & 1.57 & \cellcolor{gray!20}\textbf{24.01} & \cellcolor{gray!20}\textbf{64.32} & 4.86 & 3.37 & 1.87 \\
groovy & 1.47 & 15.84 & \cellcolor{gray!20}\textbf{65.72} & \cellcolor{gray!20}\textbf{16.06} & 0.34 & 0.90 & 1.13 \\
derby  & 2.07 & 20.52 & \cellcolor{gray!20}\textbf{48.76} & \cellcolor{gray!20}\textbf{19.48} & 7.06 & 3.48 & 0.70 \\
camel  & 0.62 & 11.16 & \cellcolor{gray!20}\textbf{67.66} & \cellcolor{gray!20}\textbf{19.10} & 1.24 & 0.23 & 0.61 \\
hbase  & 7.20 & 16.19 & \cellcolor{gray!20}\textbf{38.11} & \cellcolor{gray!20}\textbf{32.55} & 7.03 & 2.56 & 3.54 \\
jruby  & 1.24 & 14.77 & \cellcolor{gray!20}\textbf{54.02} & \cellcolor{gray!20}\textbf{19.89} & 9.55 & 1.24 & 0.53 \\
\bottomrule
\end{tabular}
\end{table}

\begin{figure}[h]
\centering
\begin{mdframed}
\begin{center}
\textbf{Naive Label-Persistent Classifier}
\end{center}

\begin{center}
\itshape
For a given file $i$ in the future version $V_{t+1}$, two cases are possible.
(1) If $i \in \mathcal{C}$, that is, the file is common and already exists in the
previous version $V_t$, the classifier assigns to $i$ the label it had in $V_t$.
(2) If $i \in \mathcal{A}$, that is, the file is newly added in $V_{t+1}$, the
classifier predicts the label \textit{benign}.
\end{center}
\end{mdframed}
\caption{Naive label-persistent classification strategy.}
\label{fig:naive_classifier}
\end{figure}

\begin{table}[!htbp]
\centering
\caption{Performance of the naive label-persistent classifier.}
\label{tab:naive_baseline}
\small
\setlength{\tabcolsep}{3pt}
\renewcommand{\arraystretch}{1.15}
\begin{tabular}{lcccccc}
\toprule
\textbf{Dataset} & Prec & Rec & F1 & AUC & FPR & MCC \\
\midrule
lucene & 0.6995 & 0.7939 & 0.7353 & 0.7939 & 0.0523 & 0.4842 \\
groovy & 0.8317 & 0.6689 & 0.7216 & 0.6689 & 0.0059 & 0.4734 \\
camel  & 0.5653 & 0.5672 & 0.5663 & 0.5672 & 0.0137 & 0.1326 \\
hbase  & 0.5985 & 0.6056 & 0.6018 & 0.6056 & 0.0772 & 0.2040 \\
jruby  & 0.5496 & 0.7289 & 0.5627 & 0.7289 & 0.1022 & 0.2132 \\
derby  & 0.6530 & 0.7994 & 0.6931 & 0.7994 & 0.0921 & 0.4280 \\
\bottomrule
\end{tabular}
\end{table}

We explore the design choices in Section~\ref{sec:experimental_setting} for the implementation of traditional baselines described in Section~\ref{sec:traditional_baselines}.
The resulting F1 scores for all configurations are reported in \autoref{tab:ml_results}. The pooling strategy with the highest F1 score on the test set is selected (detailed report in \autoref{tab:subset_macro_f1_common}). Random Forest achieves the highest F1 score on common files (0.7190), the highest score on the \textit{Unchanged} subset (0.9039), and near-best relative performance on the \textit{Changed} subset (0.2373). On this basis, we select it as the representative traditional baseline along with SDPERL\cite{hesamolhokama2024sdperl} and DBN-based feature extraction approach\cite{wang2018deep} for comparison in \autoref{tab:naive_vs_methods}.

\begin{table*}[!htbp]
\centering
\caption{F1 values across embedding and pooling configurations for traditional classifiers. 
Columns correspond to combinations of token-level and chunk-level pooling, 
while \textit{openai} denotes embeddings without token-level pooling. 
Models labeled with \textit{\_Balanced} apply \texttt{class\_weight="balanced"} to handle class imbalance. 
\textbf{Bold} values indicate the best-performing configuration per model.}
\label{tab:ml_results}
\renewcommand{\arraystretch}{1.2} 
\resizebox{\textwidth}{!}{
\begin{tabular}{lcccccccccccc}
\hline
\textbf{Model} & \textbf{cls+avg} & \textbf{cls+max} & \textbf{cls+min} & \textbf{max+avg} & \textbf{max+max} & \textbf{max+min} & \textbf{mean+avg} & \textbf{mean+max} & \textbf{mean+min} & \textbf{openai+avg} & \textbf{openai+max} & \textbf{openai+min} \\
\hline
LogisticRegression & 0.5923 & 0.6287 & 0.5923 & 0.6517 & \textbf{0.6710} & 0.6350 & 0.5920 & 0.6152 & 0.5886 & 0.5961 & 0.5961 & 0.5961 \\
LogisticRegression\_Balanced & 0.5731 & 0.6342 & 0.6283 & 0.6567 & \textbf{0.6762} & 0.6525 & 0.6002 & 0.6058 & 0.6177 & 0.6142 & 0.6241 & 0.6180 \\
RandomForest & 0.5776 & 0.6155 & 0.6018 & 0.5345 & 0.5808 & 0.5329 & 0.5525 & 0.5689 & 0.5479 & 0.7157 & \textbf{0.7210} & 0.7102 \\
RandomForest\_Balanced & 0.6179 & 0.6188 & 0.6243 & 0.5555 & 0.5452 & 0.5353 & 0.5798 & 0.5985 & 0.5948 & 0.7102 & 0.7013 & \textbf{0.7156} \\
SVM\_Linear & 0.6267 & 0.6669 & 0.6533 & 0.6766 & \textbf{0.6907} & 0.6583 & 0.6382 & 0.6369 & 0.5940 & 0.6433 & 0.6415 & 0.6415 \\
SVM\_Linear\_Balanced & 0.5847 & 0.6563 & 0.6563 & \textbf{0.6953} & 0.6783 & 0.6750 & 0.6143 & 0.6192 & 0.6323 & 0.6328 & 0.6423 & 0.6436 \\
\hline
\end{tabular}
}
\end{table*}

\begin{table}[!htbp]
\centering
\caption{F1 scores on subsets of common files ($\mathcal{C}$) grouped by file
evolution (Section~\ref{sec:problem_formulation}): \textit{Same} files have
identical source code, \textit{Unchanged} files have changed source code with
unchanged labels, and \textit{Changed} files have both source-code and label
changes.}
\label{tab:subset_macro_f1_common}
\renewcommand{\arraystretch}{1.2} 
\begin{tabular}{l c >{\color{ForestGreen}}c >{\color{BrickRed}}c}
\hline
\textbf{Model} & \textbf{Same} & \textbf{Unchanged} & \textbf{Changed} \\
\hline
LogisticRegression            & 0.8992 & \textbf{0.8077} & \textbf{0.2484} \\
LogisticRegression\_Balanced  & 0.7809 & \textbf{0.7334} & \textbf{0.1807} \\
RandomForest                  & 0.9278 & \textbf{0.9039} & \textbf{0.2373} \\
RandomForest\_Balanced        & 0.8873 & \textbf{0.5304} & \textbf{0.1539} \\
SVM\_Linear                   & 0.9278 & \textbf{0.8203} & \textbf{0.2029} \\
SVM\_Linear\_Balanced         & 0.9278 & \textbf{0.8156} & \textbf{0.1539} \\
\hline
\end{tabular}
\end{table}

Our findings expose four fundamental problems in traditional SDP evaluation leading to an \textit{illusion of accuracy}:

\begin{itemize}
\item \textbf{Leakage through file overlap and label persistence.} Successive software versions share a large majority of files, many of which retain their labels, creating a leakage that favors predictors relying on historical labels (\autoref{tab:file_evolution_partition}).
\item \textbf{Metric inflation.} Commonly used metrics in SDP, such as F1, are inflated in this setting, as demonstrated by (\autoref{fig:naive_classifier} and \autoref{tab:naive_baseline}).
\item \textbf{Performance Drop on label-changing files.} Traditional models achieve high F1 scores on status-unchanged files but perform poorly on files whose defect labels change (\autoref{tab:subset_macro_f1_common}).
\item \textbf{Marginal gains over naive baselines.} State-of-the-art models achieve performance levels that are indistinguishable from the naive baseline (\autoref{tab:naive_vs_methods}).
\end{itemize}

\begin{table*}[!htbp]
\centering
\caption{Comparison of the naive baseline and representative methods across PROMISE datasets using Precision, Recall, F1, and AUC.}
\label{tab:naive_vs_methods}
\small
\setlength{\tabcolsep}{3pt}
\renewcommand{\arraystretch}{1.15}
\begin{tabular}{lcccccccccccccccc}
\toprule
\multirow{2}{*}{\textbf{Dataset}}
& \multicolumn{4}{c}{\textbf{Naive Classifier}}
& \multicolumn{4}{c}{\textbf{Random Forest}}
& \multicolumn{4}{c}{\textbf{SDPERL}}
& \multicolumn{4}{c}{\textbf{DBN}} \\
\cmidrule(lr){2-5}
\cmidrule(lr){6-9}
\cmidrule(lr){10-13}
\cmidrule(lr){14-17}
& Prec & Rec & F1 & AUC
& Prec & Rec & F1 & AUC
& Prec & Rec & F1 & AUC
& Prec & Rec & F1 & AUC \\
\midrule
lucene
& 0.69 & \textbf{0.79} & \textbf{0.73} & 0.793
& \textbf{0.72} & 0.71 & 0.72 & \textbf{0.888}
& 0.68 & 0.77 & 0.72 & 0.884
& 0.54 & 0.54 & 0.54 & 0.539 \\
groovy
& 0.83 & 0.66 & \textbf{0.72} & 0.668      
& \textbf{0.86} & 0.59 & 0.64 & \textbf{0.812}    
& 0.72 & \textbf{0.68} & 0.70 & 0.809
& 0.70 & 0.54 & 0.57 & 0.545 \\
derby
& \textbf{0.65} & \textbf{0.79} & \textbf{0.69} & 0.799
& 0.65 & 0.75 & 0.69 & 0.869       
& 0.62 & 0.79 & 0.65 & \textbf{0.872}
& 0.59 & 0.70 & 0.61 & 0.701 \\
camel
& 0.56 & 0.56 & 0.56 & 0.567  
& \textbf{0.57} & 0.53 & 0.54 & 0.795      
& 0.57 & \textbf{0.58} & \textbf{0.58} & \textbf{0.803}
& 0.56 & 0.58 & 0.57 & 0.583 \\
hbase
& 0.59 & 0.60 & \textbf{0.60} & 0.605    
& 0.57 & 0.55 & 0.56 & \textbf{0.71}
& 0.59 & \textbf{0.62} & 0.60 & 0.702
& \textbf{0.63} & 0.53 & 0.54 & 0.533 \\
jruby
& \textbf{0.54} & \textbf{0.72} & \textbf{0.56} & 0.728
& 0.53 & 0.60 & 0.54 & \textbf{0.847}       
& -- & -- & -- & --  
& -- & -- & -- & --        \\
\bottomrule
\end{tabular}
\end{table*}

\subsection{RQ2: Change-Aware LLM-Baselines}
\label{sec:RQ2}
We evaluate the change-aware LLM baselines introduced in Section \ref{sec:llm_baselines_change} using a diverse set of proprietary and open-source models (see Section \ref{sec:experimental_setting}). \autoref{tab:llm_results} reports the F1 scores achieved across the nine prompting methods M0–M8. Method M0 represents single-version classification, while M1–M8 correspond to the change-aware reasoning strategies defined in Section \ref{sec:llm_baselines_change}. The results demonstrate the clear advantage of change-aware approach.
Among the change-aware methods, M5 (Diff-only Reasoning) stands out by focusing exclusively on code-changes (see \autoref{fig:prompt_m5}), validating the effectiveness of the problem formulation introduced in Section \ref{sec:problem_formulation}.

\begin{table*}[!htbp]
\centering
\caption{F1 scores across models (rows) and experimental methods (columns). Best result per model is highlighted in bold. The M5 column is shaded to emphasize its strong performance.}
\label{tab:llm_results}
\renewcommand{\arraystretch}{1.2}
\setlength{\tabcolsep}{5pt}
\scriptsize
\begin{tabular}{lcccccccccc}
\toprule
\textbf{Model} & \textbf{M0} & \textbf{M1} & \textbf{M2} & \textbf{M3} & \textbf{M4} & \textbf{M5} & \textbf{M6} & \textbf{M7} & \textbf{M8} \\
\midrule
DeepSeek-Reasoning
& 0.502 & 0.513 & 0.514 & 0.478 & 0.518 & \cellcolor{gray!20}\textbf{0.645} & 0.530 & 0.530 & 0.637 \\

Gemini-2.5-Flash-Lite
& 0.490 & 0.524 & 0.497 & 0.507 & 0.581 & \cellcolor{gray!20}\textbf{0.609} & 0.522 & 0.556 & 0.567 \\

GPT-5-Mini
& 0.472 & 0.484 & 0.478 & 0.489 & 0.500 & \cellcolor{gray!20}0.629 & 0.514 & 0.574 & \textbf{0.632} \\

DeepSeek-V3.1
& 0.279 & 0.424 & 0.498 & 0.512 & 0.525 & \cellcolor{gray!20}\textbf{0.664} & 0.458 & 0.594 & 0.535 \\

DeepSeek-V3.1-671B-Terminus
& 0.255 & 0.449 & 0.518 & 0.553 & 0.520 & \cellcolor{gray!20}\textbf{0.588} & 0.466 & 0.576 & 0.546 \\

Codestral-2405
& 0.472 & 0.498 & 0.530 & 0.517 & 0.547 & \cellcolor{gray!20}\textbf{0.633} & 0.486 & 0.565 & 0.591 \\

Codestral-2501
& 0.463 & 0.466 & 0.541 & 0.490 & 0.575 & \cellcolor{gray!20}\textbf{0.650} & 0.469 & 0.603 & 0.589 \\

Mistral-Small-3.1-24B-Instruct-2503
& 0.500 & 0.498 & 0.498 & 0.475 & 0.494 & \cellcolor{gray!20}\textbf{0.630} & 0.482 & 0.514 & 0.625 \\

GPT-O4-Mini-2025-04-16
& 0.491 & 0.516 & 0.538 & 0.540 & 0.536 & \cellcolor{gray!20}\textbf{0.664} & 0.510 & 0.634 & 0.651 \\

Mistral-Small-2503
& -- & -- & -- & -- & -- & \cellcolor{gray!20}\textbf{0.624} & -- & -- & 0.608 \\

GPT-5-Chat
& -- & -- & -- & -- & -- & \cellcolor{gray!20}\textbf{0.614} & -- & -- & 0.613 \\

GPT-5-Nano-2025-08-07
& 0.489 & -- & -- & -- & -- & \cellcolor{gray!20}\textbf{0.541} & -- & 0.533 & 0.527 \\

Qwen2.5-Coder-32B-Instruct
& -- & -- & -- & -- & -- & \cellcolor{gray!20}\textbf{0.701} & -- & -- & 0.622 \\

Open-Mixtral-8x7B
& -- & -- & -- & -- & -- & \cellcolor{gray!20}\textbf{0.598} & -- & -- & 0.547 \\
\bottomrule
\end{tabular}
\end{table*}
A more detailed report is given in \autoref{tab:m5_macro_f1}. While M5 achieves strong overall F1, a consistent drop is observed on the Changed subset, mirroring the trends in \ref{sec:RQ1}. 
Analysis of the top-performing models across the four change subsets reveals that change-aware reasoning addresses label-persistence bias when the previous label is \textit{defective}; however, performance on D01 remains near zero while B00 accuracy approaches 99\% (see \autoref{tab:m5_subsets}). When the previous label is \textit{benign}, bug introduction (D01) is almost impossible to detect. As a result, strong total F1 scores can still mask failures on critical subsets, leading yet to another \textit{illusion of accuracy}, though less severe than in traditional formulations (check Section \ref{sec:RQ1}).

\begin{table}[H]
\centering
\caption{Detailed F1 scores for M5 across all models from \autoref{tab:llm_results}.}
\label{tab:m5_macro_f1}
\renewcommand{\arraystretch}{1.2}
\small
\setlength{\tabcolsep}{6pt}

\resizebox{\columnwidth}{!}{
\begin{tabular}{l c c c}
\hline
\textbf{Model} & \textbf{Total} & \textbf{Unchanged} & \textbf{Changed} \\
\hline
Codestral-2405 & 0.6334 & 0.7315 & 0.2685 \\
Codestral-2501 & 0.6503 & 0.7684 & 0.2484 \\
DeepSeek-Reasoning & 0.6449 & 0.7325 & 0.3836 \\
DeepSeek-V3.1 & 0.6637 & 0.8249 & 0.2105 \\
DeepSeek-V3.1-671B-Terminus & 0.5880 & 0.6825 & 0.2241 \\
Gemini-2.5-Flash-Lite & 0.6094 & 0.7042 & 0.2969 \\
GPT-5-Chat & 0.6144 & 0.6948 & 0.3576 \\
GPT-5-Mini & 0.6291 & 0.7052 & 0.3841 \\
GPT-5-Nano-2025-08-07 & 0.5414 & 0.5602 & 0.4299 \\
GPT-O4-Mini-2025-04-16 & 0.6637 & 0.8478 & 0.1739 \\
Mistral-Small-2503 & 0.6244 & 0.6944 & 0.4000 \\
Mistral-Small-3.1-24B-Instruct-2503 & 0.6297 & 0.7028 & 0.3960 \\
Open-Mixtral-8x7B & 0.5979 & 0.6329 & \textbf{0.4592} \\
Qwen2.5-Coder-32B-Instruct & \textbf{0.7014} & \textbf{0.8526} & 0.3130 \\
\hline
\end{tabular}
}
\end{table}

\begin{table}[H] 
\centering 
\caption{Representative subset accuracies (method M5). 
} 
\label{tab:m5_subsets} 
\scriptsize 
\setlength{\tabcolsep}{5pt} 
\renewcommand{\arraystretch}{1.2} 
\begin{tabular}{@{}l c >{\columncolor{red!20}}c >{\columncolor{gray!20}}c >{\columncolor{gray!20}}c c@{}}

\toprule 
Model & B00 & D01 & B10 & D11 & F1 \\ 
\midrule

Qwen2.5-Coder-32B 
& 99.4\% & \textbf{0.0\%} & 63.1\% & 62.2\% & {0.701} \\

GPT-O4-Mini-2025-04-16 
& 99.3\% & \textbf{4.0\%} & 26.2\% & 62.2\% & {0.664} \\

DeepSeek-V3.1 
& 98.5\% & \textbf{0.0\%} & 36.9\% & 64.4\% & {0.664} \\

Codestral-2501 
& 96.6\% & \textbf{4.0\%} & 41.5\% & 64.4\% & {0.650} \\

DeepSeek-Reasoning 
& 99.1\% & \textbf{0.0\%} & 86.2\% & 37.8\% & {0.645} \\

GPT-5-Mini 
& 98.1\% & \textbf{4.0\%} & 76.9\% & 37.8\% & {0.629} \\

\bottomrule 
\end{tabular} 
\end{table}

\subsection{RQ3: Proposed Multi-Agent Debate Framework }
\label{sec:RQ3}
We select four LLMs as candidates to fill the roles and report the performance across all 52 model–role combinations in \autoref{tab:all-combinations}. 
To select a single configuration, we first discard strictly dominated candidates, leaving 37 non-dominated configurations. Applying $Acc(B00) > 0.5$ and $\text{HMD}, \text{HMB} > 0.4$ further reduces this set to three candidates in \autoref{tab:all-combinations}, highlighted in blue. 
Finally, we choose \hyperlink{case3}{Combination~3} due to its near-dominance and fewer false positives (higher $B00$ performance), which is particularly important given that $B00$ is relatively large in scale. Moreover, we show performance sensitivity to \texttt{max\_lines} parameter in \autoref{tab:context-sensitivity}. We select $\texttt{max\_lines} = 600$ for final results in Section \ref{sec:final_comparison}.

\begin{table*}[!htbp]
\centering
\caption{
Performance across 52 model--role combinations. Model candidates are Gemini-2.5-Flash-Lite (Gemini), GPT-5-Mini (GPT), Codestral-2501 (Codestral), and DeepSeek-V3.1 (DeepSeek). \texttt{max\_lines} is $400$ with $\texttt{debate\_rounds} = 1$.
Evaluation uses 100\% of the B10, D01, and D11 subsets, and 100 randomly sampled instances from B00; for each metric column, cell colors are independently scaled to reflect relative performance per metric, with darker shading indicating more extreme values.
}
\label{tab:all-combinations}

\renewcommand{\arraystretch}{1.3} 
\resizebox{0.90\textwidth}{!}{
\begin{tabular}{
>{\centering\arraybackslash}m{2.5cm} 
>{\centering\arraybackslash}m{2.0cm} 
>{\centering\arraybackslash}m{2.0cm} 
>{\centering\arraybackslash}m{2.0cm} 
>{\centering\arraybackslash}m{2.0cm} 
>{\centering\arraybackslash}m{0.7cm} 
>{\centering\arraybackslash}m{0.7cm} 
>{\centering\arraybackslash}m{0.7cm} 
>{\centering\arraybackslash}m{0.7cm} 
>{\centering\arraybackslash}m{1.2cm} 
>{\centering\arraybackslash}m{1.2cm} 
>{\centering\arraybackslash}m{1.2cm} 
>{\centering\arraybackslash}m{1.2cm} 
>{\centering\arraybackslash}m{1.2cm} 
}
\toprule
\multirow{2}{*}{Group} & \multirow{2}{*}{Analyzer} & \multirow{2}{*}{Proposer} & \multirow{2}{*}{Skeptic} & \multirow{2}{*}{Judge} &
\multicolumn{4}{c}{Subset Accuracy} &
\multicolumn{2}{c}{Change-aware Metrics} &
\multicolumn{3}{c}{F1 Score} \\
\cmidrule(lr){6-9} \cmidrule(lr){10-11} \cmidrule(lr){12-14}
& & & & & B00 & D01 & D11 & B10 & HMD & HMB & Changed & Unchanged & Total \\
\midrule

\multirow{4}{*}{All roles same model}
& GPT & GPT & GPT & GPT & \cellcolor{orange!12}0.51 & \cellcolor{orange!18}0.44 & \cellcolor{green!11}0.80 & \cellcolor{red!14}0.14 & \cellcolor{red!10}0.24 & \cellcolor{green!12}0.47 & \cellcolor{red!13}0.22 & \cellcolor{orange!10}0.60 & \cellcolor{orange!15}0.45 \\
& DeepSeek & DeepSeek & DeepSeek & DeepSeek & \cellcolor{orange!13}0.49 & \cellcolor{orange!13}0.60 & \cellcolor{orange!17}0.69 & \cellcolor{orange!10}0.35 & \cellcolor{green!14}0.47 & \cellcolor{green!17}0.54 & \cellcolor{orange!10}0.42 & \cellcolor{orange!14}0.54 & \cellcolor{orange!10}0.50 \\
& Codestral & Codestral & Codestral & Codestral & \cellcolor{green!14}0.67 & \cellcolor{red!14}0.28 & \cellcolor{green!19}0.91 & \cellcolor{red!18}0.08 & \cellcolor{red!16}0.14 & \cellcolor{orange!12}0.40 & \cellcolor{red!19}0.13 & \cellcolor{green!16}0.74 & \cellcolor{green!10}0.51 \\
& Gemini & Gemini & Gemini & Gemini & \cellcolor{red!10}0.31 & \cellcolor{green!10}0.72 & \cellcolor{orange!17}0.69 & \cellcolor{orange!10}0.34 & \cellcolor{green!13}0.45 & \cellcolor{orange!10}0.43 & \cellcolor{green!11}0.44 & \cellcolor{red!10}0.43 & \cellcolor{orange!17}0.43 \\
\midrule

\multirow{24}{*}[2.0\baselineskip]{%
\makecell{Same model for \\Skeptic and Proposer,\\ different models for\\Analyzer and Judge
}
}
& DeepSeek & GPT & GPT & Codestral & \cellcolor{orange!16}0.42 & \cellcolor{orange!12}0.64 & \cellcolor{green!14}0.84 & \cellcolor{red!15}0.12 & \cellcolor{red!12}0.21 & \cellcolor{green!15}0.51 & \cellcolor{red!10}0.26 & \cellcolor{orange!13}0.55 & \cellcolor{orange!16}0.44 \\
& DeepSeek & GPT & GPT & Gemini & \cellcolor{orange!18}0.36 & \cellcolor{orange!16}0.52 & \cellcolor{orange!15}0.71 & \cellcolor{orange!12}0.31 & \cellcolor{green!12}0.43 & \cellcolor{orange!10}0.43 & \cellcolor{orange!14}0.36 & \cellcolor{orange!18}0.47 & \cellcolor{orange!17}0.43 \\
& Codestral & GPT & GPT & DeepSeek & \cellcolor{orange!18}0.36 & \cellcolor{orange!12}0.64 & \cellcolor{green!14}0.84 & \cellcolor{orange!14}0.29 & \cellcolor{green!12}0.43 & \cellcolor{green!11}0.46 & \cellcolor{orange!12}0.39 & \cellcolor{orange!15}0.51 & \cellcolor{orange!13}0.47 \\
& Codestral & GPT & GPT & Gemini & \cellcolor{red!10}0.31 & \cellcolor{orange!16}0.52 & \cellcolor{red!14}0.60 & \cellcolor{orange!10}0.34 & \cellcolor{green!12}0.43 & \cellcolor{orange!13}0.39 & \cellcolor{orange!12}0.38 & \cellcolor{red!12}0.40 & \cellcolor{red!11}0.39 \\
& Gemini & GPT & GPT & DeepSeek & \cellcolor{orange!16}0.40 & \cellcolor{red!14}0.28 & \cellcolor{green!13}0.82 & \cellcolor{red!15}0.12 & \cellcolor{red!12}0.21 & \cellcolor{orange!17}0.33 & \cellcolor{red!16}0.17 & \cellcolor{orange!14}0.53 & \cellcolor{red!11}0.39 \\
& Gemini & GPT & GPT & Codestral & \cellcolor{orange!19}0.35 & \cellcolor{orange!18}0.44 & \cellcolor{green!13}0.82 & \cellcolor{red!16}0.11 & \cellcolor{red!13}0.19 & \cellcolor{orange!13}0.39 & \cellcolor{red!14}0.20 & \cellcolor{orange!16}0.50 & \cellcolor{red!12}0.38 \\
\addlinespace[1.5ex]

& GPT & DeepSeek & DeepSeek & Codestral & \cellcolor{red!11}0.29 & \cellcolor{green!14}0.84 & \cellcolor{orange!10}0.78 & \cellcolor{red!10}0.19 & \cellcolor{orange!16}0.30 & \cellcolor{orange!10}0.43 & \cellcolor{orange!14}0.36 & \cellcolor{orange!19}0.44 & \cellcolor{orange!19}0.41 \\
& GPT & DeepSeek & DeepSeek & Gemini & \cellcolor{red!14}0.23 & \cellcolor{green!13}0.80 & \cellcolor{red!11}0.64 & \cellcolor{orange!18}0.23 & \cellcolor{orange!13}0.34 & \cellcolor{orange!15}0.36 & \cellcolor{orange!12}0.39 & \cellcolor{red!14}0.36 & \cellcolor{red!13}0.37 \\
& Codestral & DeepSeek & DeepSeek & GPT & \cellcolor{red!10}0.32 & \cellcolor{orange!10}0.68 & \cellcolor{red!12}0.62 & \cellcolor{orange!11}0.33 & \cellcolor{green!12}0.43 & \cellcolor{green!10}0.44 & \cellcolor{green!10}0.43 & \cellcolor{red!11}0.41 & \cellcolor{orange!18}0.42 \\
& Codestral & DeepSeek & DeepSeek & Gemini & \cellcolor{red!10}0.31 & \cellcolor{orange!13}0.60 & \cellcolor{red!12}0.62 & \cellcolor{green!11}0.38 & \cellcolor{green!15}0.48 & \cellcolor{orange!11}0.41 & \cellcolor{green!11}0.44 & \cellcolor{red!11}0.41 & \cellcolor{orange!18}0.42 \\
& Gemini & DeepSeek & DeepSeek & GPT & \cellcolor{orange!19}0.34 & \cellcolor{orange!16}0.52 & \cellcolor{orange!18}0.68 & \cellcolor{orange!12}0.31 & \cellcolor{green!12}0.43 & \cellcolor{orange!11}0.41 & \cellcolor{orange!13}0.37 & \cellcolor{orange!19}0.44 & \cellcolor{orange!18}0.42 \\
& Gemini & DeepSeek & DeepSeek & Codestral & \cellcolor{red!14}0.22 & \cellcolor{orange!18}0.44 & \cellcolor{orange!15}0.71 & \cellcolor{red!15}0.12 & \cellcolor{red!12}0.21 & \cellcolor{red!10}0.29 & \cellcolor{red!14}0.21 & \cellcolor{red!13}0.37 & \cellcolor{red!20}0.31 \\
\addlinespace[1.5ex]

\hypertarget{case1}{\textbf{(1)}} & \cellcolor{blue!20}GPT & \cellcolor{blue!20}Codestral & \cellcolor{blue!20}Codestral & \cellcolor{blue!20}DeepSeek & \cellcolor{green!18}0.78 & \cellcolor{red!11}0.36 & \cellcolor{orange!17}0.69 & \cellcolor{orange!14}0.29 & \cellcolor{green!11}0.41 & \cellcolor{green!13}0.49 & \cellcolor{orange!18}0.30 & \cellcolor{green!15}0.72 & \cellcolor{green!14}0.55 \label{Case:1}\\
\hypertarget{case2}{\textbf{(2)}} & \cellcolor{blue!20}GPT & \cellcolor{blue!20}Codestral & \cellcolor{blue!20}Codestral & \cellcolor{blue!20}Gemini & \cellcolor{green!15}0.70 & \cellcolor{red!13}0.32 & \cellcolor{orange!17}0.69 & \cellcolor{green!16}0.46 & \cellcolor{green!19}0.55 & \cellcolor{green!10}0.44 & \cellcolor{orange!12}0.39 & \cellcolor{green!12}0.67 & \cellcolor{green!15}0.56 \label{Case:2}\\
& DeepSeek & Codestral & Codestral & GPT & \cellcolor{green!17}0.76 & \cellcolor{red!14}0.28 & \cellcolor{orange!17}0.69 & \cellcolor{orange!18}0.22 & \cellcolor{orange!14}0.33 & \cellcolor{orange!11}0.41 & \cellcolor{red!12}0.23 & \cellcolor{green!15}0.71 & \cellcolor{green!11}0.52 \\
& DeepSeek & Codestral & Codestral & Gemini & \cellcolor{green!15}0.70 & \cellcolor{red!14}0.28 & \cellcolor{red!14}0.60 & \cellcolor{green!14}0.42 & \cellcolor{green!16}0.49 & \cellcolor{orange!12}0.40 & \cellcolor{orange!14}0.35 & \cellcolor{green!11}0.64 & \cellcolor{green!11}0.52 \\
& Gemini & Codestral & Codestral & GPT & \cellcolor{green!14}0.67 & \cellcolor{red!20}0.12 & \cellcolor{orange!11}0.76 & \cellcolor{orange!15}0.27 & \cellcolor{orange!10}0.39 & \cellcolor{red!17}0.20 & \cellcolor{red!14}0.21 & \cellcolor{green!13}0.68 & \cellcolor{orange!10}0.50 \\
& Gemini & Codestral & Codestral & DeepSeek & \cellcolor{green!16}0.73 & \cellcolor{red!13}0.32 & \cellcolor{red!12}0.62 & \cellcolor{orange!18}0.22 & \cellcolor{orange!14}0.32 & \cellcolor{green!10}0.44 & \cellcolor{red!12}0.24 & \cellcolor{green!12}0.67 & \cellcolor{orange!10}0.50 \\
\addlinespace[1.5ex]

& GPT & Gemini & Gemini & DeepSeek & \cellcolor{red!11}0.30 & \cellcolor{green!16}0.91 & \cellcolor{orange!15}0.71 & \cellcolor{orange!14}0.29 & \cellcolor{green!11}0.41 & \cellcolor{green!11}0.46 & \cellcolor{green!11}0.45 & \cellcolor{red!10}0.43 & \cellcolor{orange!16}0.44 \\
& GPT & Gemini & Gemini & Codestral & \cellcolor{red!20}0.09 & \cellcolor{green!18}0.96 & \cellcolor{orange!15}0.71 & \cellcolor{orange!12}0.32 & \cellcolor{green!12}0.44 & \cellcolor{red!20}0.16 & \cellcolor{green!14}0.50 & \cellcolor{red!20}0.26 & \cellcolor{red!14}0.36 \\
& DeepSeek & Gemini & Gemini & GPT & \cellcolor{red!10}0.32 & \cellcolor{green!13}0.80 & \cellcolor{red!14}0.60 & \cellcolor{green!11}0.37 & \cellcolor{green!13}0.45 & \cellcolor{green!11}0.46 & \cellcolor{green!14}0.49 & \cellcolor{red!11}0.41 & \cellcolor{orange!16}0.44 \\
& DeepSeek & Gemini & Gemini & Codestral & \cellcolor{red!17}0.16 & \cellcolor{green!15}0.88 & \cellcolor{orange!14}0.73 & \cellcolor{orange!12}0.32 & \cellcolor{green!13}0.45 & \cellcolor{red!12}0.27 & \cellcolor{green!13}0.48 & \cellcolor{red!16}0.33 & \cellcolor{red!11}0.39 \\
& Codestral & Gemini & Gemini & GPT & \cellcolor{red!15}0.20 & \cellcolor{green!15}0.88 & \cellcolor{red!14}0.60 & \cellcolor{green!17}0.47 & \cellcolor{green!18}0.53 & \cellcolor{orange!17}0.33 & \cellcolor{green!19}0.58 & \cellcolor{red!16}0.32 & \cellcolor{orange!18}0.42 \\
& Codestral & Gemini & Gemini & DeepSeek & \cellcolor{red!10}0.31 & \cellcolor{green!11}0.76 & \cellcolor{red!11}0.64 & \cellcolor{orange!10}0.34 & \cellcolor{green!12}0.44 & \cellcolor{green!10}0.44 & \cellcolor{green!11}0.45 & \cellcolor{red!11}0.41 & \cellcolor{orange!17}0.43 \\
\addlinespace[1.5ex]
\midrule

\multirow{24}{*}{All roles different}
& GPT & DeepSeek & Codestral & Gemini & \cellcolor{green!13}0.64 & \cellcolor{red!17}0.20 & \cellcolor{red!12}0.62 & \cellcolor{orange!10}0.35 & \cellcolor{green!13}0.45 & \cellcolor{orange!19}0.30 & \cellcolor{orange!19}0.28 & \cellcolor{orange!10}0.61 & \cellcolor{orange!12}0.48 \\
& GPT & DeepSeek & Gemini & Codestral & \cellcolor{red!18}0.13 & \cellcolor{green!19}1.00 & \cellcolor{green!19}0.91 & \cellcolor{red!13}0.15 & \cellcolor{orange!18}0.26 & \cellcolor{red!14}0.23 & \cellcolor{orange!13}0.37 & \cellcolor{red!14}0.35 & \cellcolor{red!14}0.36 \\
& GPT & Codestral & DeepSeek & Gemini & \cellcolor{red!12}0.26 & \cellcolor{orange!12}0.64 & \cellcolor{green!11}0.80 & \cellcolor{orange!12}0.32 & \cellcolor{green!14}0.46 & \cellcolor{orange!14}0.37 & \cellcolor{orange!10}0.41 & \cellcolor{red!11}0.42 & \cellcolor{orange!18}0.42 \\
& GPT & Codestral & Gemini & DeepSeek & \cellcolor{red!11}0.29 & \cellcolor{green!12}0.79 & \cellcolor{green!11}0.80 & \cellcolor{orange!18}0.23 & \cellcolor{orange!12}0.36 & \cellcolor{orange!11}0.42 & \cellcolor{orange!12}0.38 & \cellcolor{orange!19}0.45 & \cellcolor{orange!18}0.42 \\
& GPT & Gemini & DeepSeek & Codestral & \cellcolor{red!10}0.31 & \cellcolor{green!14}0.84 & \cellcolor{green!14}0.84 & \cellcolor{orange!10}0.35 & \cellcolor{green!16}0.50 & \cellcolor{green!10}0.45 & \cellcolor{green!14}0.49 & \cellcolor{orange!18}0.47 & \cellcolor{orange!12}0.48 \\
& GPT & Gemini & Codestral & DeepSeek & \cellcolor{green!16}0.72 & \cellcolor{red!13}0.32 & \cellcolor{green!16}0.87 & \cellcolor{red!20}0.05 & \cellcolor{red!20}0.09 & \cellcolor{green!10}0.44 & \cellcolor{red!20}0.12 & \cellcolor{green!17}0.75 & \cellcolor{green!10}0.51 \\
\addlinespace[1.5ex]

& DeepSeek & GPT & Codestral & Gemini & \cellcolor{green!17}0.76 & \cellcolor{red!15}0.24 & \cellcolor{red!20}0.53 & \cellcolor{green!19}0.51 & \cellcolor{green!18}0.52 & \cellcolor{orange!15}0.36 & \cellcolor{orange!12}0.38 & \cellcolor{green!11}0.64 & \cellcolor{green!13}0.54 \\
& DeepSeek & GPT & Gemini & Codestral & \cellcolor{red!17}0.16 & \cellcolor{green!11}0.76 & \cellcolor{green!14}0.84 & \cellcolor{orange!16}0.26 & \cellcolor{green!10}0.40 & \cellcolor{red!12}0.26 & \cellcolor{orange!11}0.40 & \cellcolor{red!14}0.36 & \cellcolor{red!12}0.38 \\
& DeepSeek & Codestral & GPT & Gemini & \cellcolor{orange!17}0.38 & \cellcolor{orange!16}0.52 & \cellcolor{green!11}0.80 & \cellcolor{green!12}0.40 & \cellcolor{green!18}0.53 & \cellcolor{green!10}0.44 & \cellcolor{orange!10}0.42 & \cellcolor{orange!15}0.51 & \cellcolor{orange!12}0.48 \\
& DeepSeek & Codestral & Gemini & GPT & \cellcolor{red!13}0.24 & \cellcolor{green!13}0.80 & \cellcolor{green!16}0.86 & \cellcolor{orange!15}0.27 & \cellcolor{green!11}0.41 & \cellcolor{orange!14}0.38 & \cellcolor{orange!10}0.42 & \cellcolor{red!10}0.43 & \cellcolor{orange!18}0.42 \\
& DeepSeek & Gemini & GPT & Codestral & \cellcolor{orange!17}0.39 & \cellcolor{orange!14}0.56 & \cellcolor{green!18}0.89 & \cellcolor{red!18}0.08 & \cellcolor{red!16}0.14 & \cellcolor{green!11}0.46 & \cellcolor{red!14}0.20 & \cellcolor{orange!14}0.54 & \cellcolor{orange!18}0.42 \\
& DeepSeek & Gemini & Codestral & GPT & \cellcolor{green!16}0.73 & \cellcolor{orange!17}0.48 & \cellcolor{green!18}0.89 & \cellcolor{red!20}0.05 & \cellcolor{red!20}0.09 & \cellcolor{green!19}0.58 & \cellcolor{red!17}0.16 & \cellcolor{green!18}0.77 & \cellcolor{green!13}0.54 \\
\addlinespace[1.5ex]

& Codestral & GPT & DeepSeek & Gemini & \cellcolor{orange!19}0.35 & \cellcolor{green!13}0.80 & \cellcolor{red!15}0.59 & \cellcolor{orange!12}0.32 & \cellcolor{green!11}0.42 & \cellcolor{green!13}0.49 & \cellcolor{green!12}0.46 & \cellcolor{red!11}0.42 & \cellcolor{orange!17}0.43 \\
& Codestral & GPT & Gemini & DeepSeek & \cellcolor{orange!19}0.34 & \cellcolor{green!10}0.72 & \cellcolor{green!14}0.84 & \cellcolor{orange!15}0.27 & \cellcolor{green!11}0.41 & \cellcolor{green!12}0.47 & \cellcolor{orange!11}0.40 & \cellcolor{orange!16}0.50 & \cellcolor{orange!14}0.46 \\
& Codestral & DeepSeek & GPT & Gemini & \cellcolor{orange!19}0.33 & \cellcolor{orange!18}0.44 & \cellcolor{orange!18}0.67 & \cellcolor{orange!12}0.32 & \cellcolor{green!12}0.44 & \cellcolor{orange!14}0.38 & \cellcolor{orange!14}0.35 & \cellcolor{red!10}0.43 & \cellcolor{red!10}0.40 \\
& Codestral & DeepSeek & Gemini & GPT & \cellcolor{red!16}0.17 & \cellcolor{green!11}0.76 & \cellcolor{orange!14}0.73 & \cellcolor{orange!10}0.34 & \cellcolor{green!14}0.46 & \cellcolor{red!11}0.28 & \cellcolor{green!11}0.45 & \cellcolor{red!15}0.34 & \cellcolor{red!12}0.38 \\
& Codestral & Gemini & GPT & DeepSeek & \cellcolor{orange!17}0.39 & \cellcolor{orange!16}0.52 & \cellcolor{green!14}0.84 & \cellcolor{red!12}0.16 & \cellcolor{orange!18}0.26 & \cellcolor{green!10}0.45 & \cellcolor{red!10}0.26 & \cellcolor{orange!14}0.54 & \cellcolor{orange!17}0.43 \\
& Codestral & Gemini & DeepSeek & GPT & \cellcolor{orange!19}0.33 & \cellcolor{orange!12}0.64 & \cellcolor{green!14}0.84 & \cellcolor{orange!13}0.30 & \cellcolor{green!12}0.44 & \cellcolor{green!10}0.44 & \cellcolor{orange!12}0.39 & \cellcolor{orange!17}0.49 & \cellcolor{orange!15}0.45 \\
\addlinespace[1.5ex]

& Gemini & GPT & DeepSeek & Codestral & \cellcolor{orange!16}0.40 & \cellcolor{orange!13}0.60 & \cellcolor{orange!10}0.78 & \cellcolor{red!13}0.15 & \cellcolor{orange!18}0.26 & \cellcolor{green!12}0.48 & \cellcolor{orange!19}0.28 & \cellcolor{orange!15}0.52 & \cellcolor{orange!17}0.43 \\
\hypertarget{case3}{\textbf{(3)}} & \cellcolor{blue!20}Gemini & \cellcolor{blue!20}GPT & \cellcolor{blue!20}Codestral & \cellcolor{blue!20}DeepSeek & \cellcolor{green!19}0.81 & \cellcolor{red!12}0.33 & \cellcolor{green!13}0.82 & \cellcolor{orange!12}0.31 & \cellcolor{green!13}0.45 & \cellcolor{green!12}0.47 & \cellcolor{orange!18}0.30 & \cellcolor{green!19}0.80 & \cellcolor{green!19}0.60 \label{Case:3}\\
& Gemini & DeepSeek & GPT & Codestral & \cellcolor{orange!19}0.33 & \cellcolor{orange!16}0.52 & \cellcolor{orange!11}0.76 & \cellcolor{orange!16}0.25 & \cellcolor{orange!11}0.37 & \cellcolor{orange!12}0.40 & \cellcolor{orange!16}0.32 & \cellcolor{orange!18}0.46 & \cellcolor{orange!19}0.41 \\
& Gemini & DeepSeek & Codestral & GPT & \cellcolor{green!16}0.72 & \cellcolor{red!10}0.40 & \cellcolor{red!12}0.63 & \cellcolor{orange!18}0.22 & \cellcolor{orange!14}0.32 & \cellcolor{green!15}0.51 & \cellcolor{red!10}0.27 & \cellcolor{green!12}0.66 & \cellcolor{orange!10}0.50 \\
& Gemini & Codestral & GPT & DeepSeek & \cellcolor{orange!16}0.42 & \cellcolor{red!11}0.36 & \cellcolor{orange!10}0.78 & \cellcolor{red!14}0.14 & \cellcolor{red!10}0.24 & \cellcolor{orange!13}0.39 & \cellcolor{red!14}0.20 & \cellcolor{orange!14}0.53 & \cellcolor{orange!19}0.41 \\
& Gemini & Codestral & DeepSeek & GPT & \cellcolor{orange!14}0.45 & \cellcolor{orange!13}0.60 & \cellcolor{orange!14}0.73 & \cellcolor{orange!14}0.29 & \cellcolor{green!11}0.42 & \cellcolor{green!15}0.52 & \cellcolor{orange!12}0.38 & \cellcolor{orange!14}0.54 & \cellcolor{orange!12}0.48 \\
\addlinespace[1.5ex]

\bottomrule
\end{tabular}
}
\end{table*}

\begin{table*}[!htbp]
\centering
\scriptsize
\setlength{\tabcolsep}{1.8pt}
\caption{
Effect of \texttt{max\_lines} on performance for the selected configuration.
Values are reported as mean $\pm$ 95\% confidence interval over ten independent runs.
Superscript $^{*}$ denotes statistically significant differences against $MaxLines = 0$ for the corresponding metric (Wilcoxon signed-rank test~\cite{Wilcoxonsigned}, $p < 0.05$).
}
\resizebox{\textwidth}{!}{%
\begin{tabular}{c
cccc
cc
ccc}
\toprule
\multirow{2}{*}{Max} &
\multicolumn{4}{c}{Subset Acc.} &
\multicolumn{2}{c}{Change-aware Metrics} &
\multicolumn{3}{c}{F1} \\
\cmidrule(lr){2-5}
\cmidrule(lr){6-7}
\cmidrule(lr){8-10}
Lines & B00 & B10 & D01 & D11 & HMB & HMD & Changed & Unchanged & Total \\
\midrule
\cellcolor{red!20}0    & 0.76 $\pm$ 0.03 & 0.29 $\pm$ 0.03 & 0.31 $\pm$ 0.06 & 0.66 $\pm$ 0.05 & 0.43 $\pm$ 0.05 & 0.40 $\pm$ 0.03 & 0.28 $\pm$ 0.03 & 0.70 $\pm$ 0.03 & 0.54 $\pm$ 0.01 \\

100  & 0.76 $\pm$ 0.03 & 0.28 $\pm$ 0.03 & 0.36 $\pm$ 0.07 & 0.72 $\pm$ 0.04 & 0.48 $\pm$ 0.05 & 0.40 $\pm$ 0.03 & 0.29 $\pm$ 0.03 & 0.72 $\pm$ 0.02 & 0.55 $\pm$ 0.01 \\

200  & 0.77 $\pm$ 0.04 & 0.27 $\pm$ 0.04 & 0.32 $\pm$ 0.06 & 0.67 $\pm$ 0.02 & 0.44 $\pm$ 0.06 & 0.38 $\pm$ 0.04 & 0.28 $\pm$ 0.04 & 0.71 $\pm$ 0.03 & 0.54 $\pm$ 0.02 \\

400  & 0.78 $\pm$ 0.03 & 0.31 $\pm$ 0.04 & 0.29 $\pm$ 0.05 & 0.69 $\pm$ 0.05 & 0.41 $\pm$ 0.05 & 0.43 $\pm$ 0.04 & 0.29 $\pm$ 0.03 & 0.72 $\pm$ 0.02 & 0.55 $\pm$ 0.02 \\

\cellcolor{green!20}600  & 0.77 $\pm$ 0.03 & 0.32 $\pm$ 0.04 & \cellcolor{gray!20}0.41 $\pm$ 0.05$^{*}$ & 0.69 $\pm$ 0.04 & \cellcolor{gray!20}0.52 $\pm$ 0.06$^{*}$ & 0.43 $\pm$ 0.04 & \cellcolor{gray!20}0.33 $\pm$ 0.03$^{*}$ & 0.71 $\pm$ 0.02 & \cellcolor{gray!20}0.56 $\pm$ 0.02$^{*}$ \\

800  & 0.72 $\pm$ 0.05 & 0.29 $\pm$ 0.04 & 0.35 $\pm$ 0.06 & 0.70 $\pm$ 0.04 & 0.46 $\pm$ 0.05 & 0.41 $\pm$ 0.04 & 0.30 $\pm$ 0.03 & 0.69 $\pm$ 0.04 & 0.54 $\pm$ 0.02 \\

1000 & 0.73 $\pm$ 0.04 & 0.29 $\pm$ 0.02 & 0.36 $\pm$ 0.07 & \cellcolor{gray!20}0.72 $\pm$ 0.04$^{*}$ & 0.47 $\pm$ 0.06 & 0.41 $\pm$ 0.02 & 0.30 $\pm$ 0.02 & 0.71 $\pm$ 0.02 & 0.55 $\pm$ 0.01 \\
\bottomrule
\end{tabular}
}

\label{tab:context-sensitivity}
\end{table*}


\section{Discussion}
\label{sec:discussion}

\begin{figure*}[!htbp]
\centering
\includegraphics[width=\linewidth]{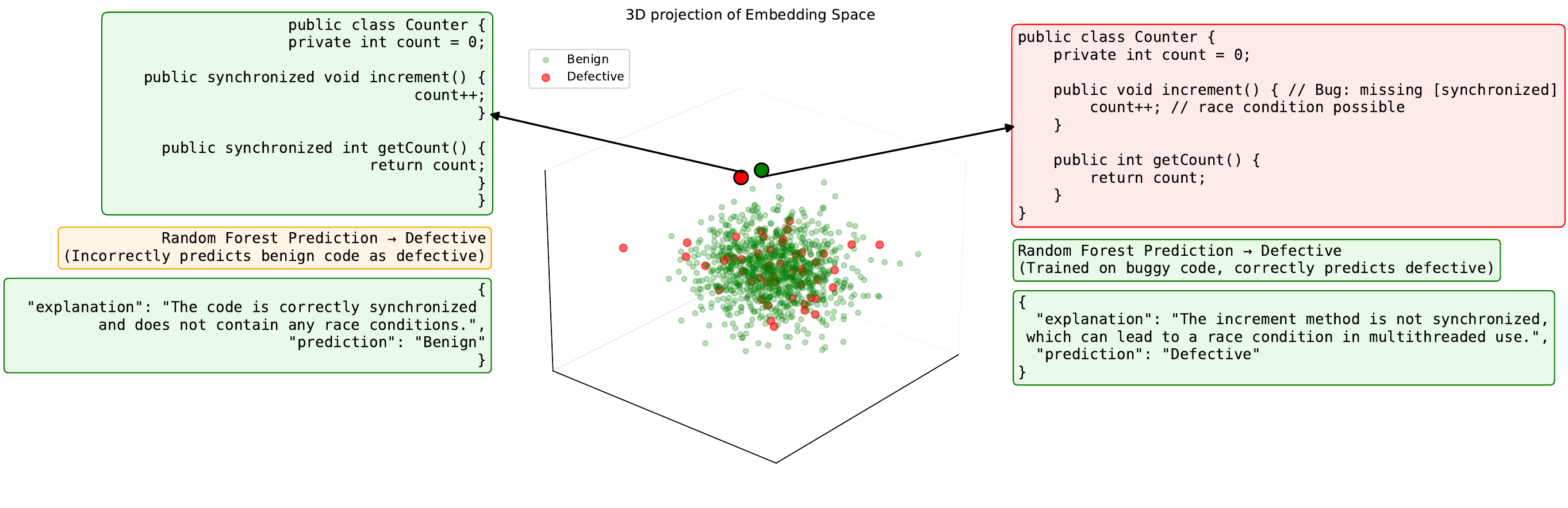}
\caption{3D projection of embeddings for benign and defective variants of a source code. Despite a race-condition edit, both map nearby, showing why embedding-based models trained on historical files can misclassify, while even a simple single-version classification by LLMs (see \autoref{fig:prompt_m0}) correctly predicts both instances.}
\label{fig:embedding_projection}
\end{figure*}

\subsection{Rethinking Traditional File-Level SDP: Why Illusion of Accuracy Occurs}
\label{sec:Rethinking_SDP}
Let $(x_1, x_2)$ be an evolution pair with defect labels $(y_1, y_2)$. Consider a change $x_2 = f_\epsilon(x_1)$ that flips the defect label while introducing only a minor edit. Traditional SDP models assume that such defect-relevant changes induce sufficiently large movements in embedding space to cross the learned decision boundary. However, as illustrated in \autoref{fig:embedding_projection}, even changes that introduce critical defects can often result in negligible embedding shifts. We show the most critical label-transition scenarios across file evolution in \autoref{fig:multi_fig}.


\begin{figure*}[!htbp]

\centering

\begin{minipage}{0.2375\textwidth}
    \includegraphics[width=\textwidth]{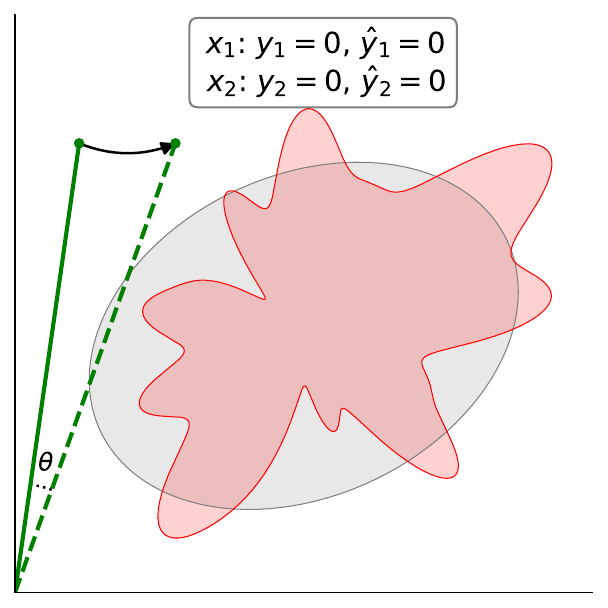}
    \centering \raisebox{0pt}{\textbf{(a)}} 
\end{minipage}\hfill
\begin{minipage}{0.2375\textwidth}
    \includegraphics[width=\textwidth]{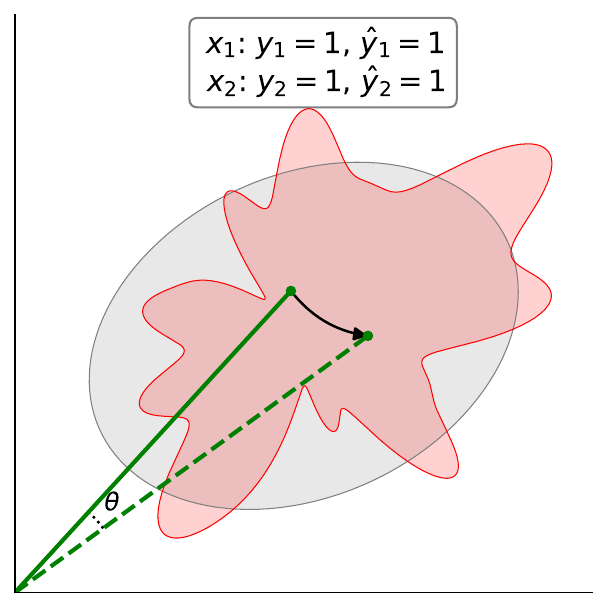}
    \centering \raisebox{0pt}{\textbf{(b)}}
\end{minipage}\hfill
\begin{minipage}{0.2375\textwidth}
    \includegraphics[width=\textwidth]{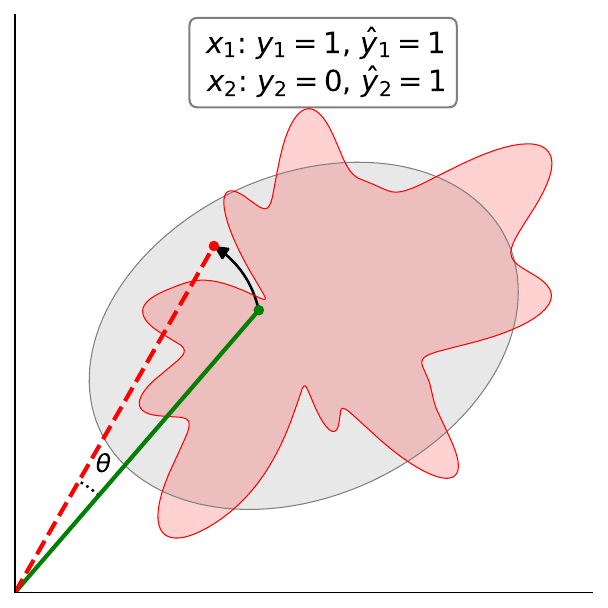}
    \centering \raisebox{0pt}{\textbf{(c)}}
\end{minipage}\hfill
\begin{minipage}{0.2375\textwidth}
    \includegraphics[width=\textwidth]{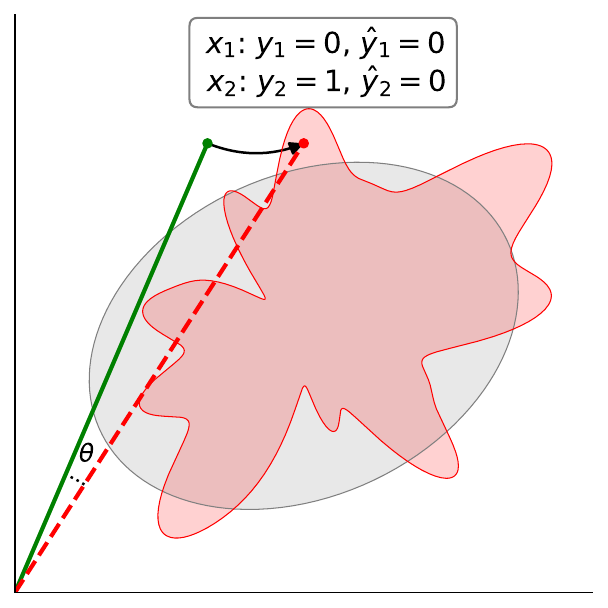}
    \centering \raisebox{0pt}{\textbf{(d)}}
\end{minipage}

\vspace{0.2cm}

\begin{minipage}{0.2375\textwidth}
    \includegraphics[width=\textwidth]{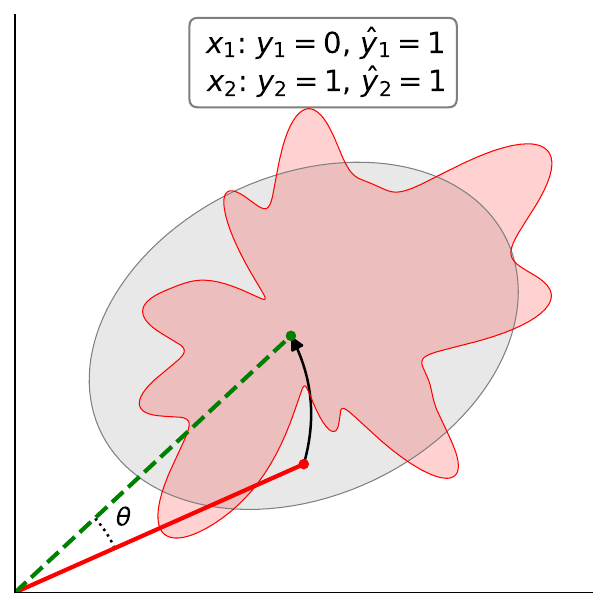}
    \centering \raisebox{0pt}{\textbf{(e)}}
\end{minipage}\hfill
\begin{minipage}{0.2375\textwidth}
    \includegraphics[width=\textwidth]{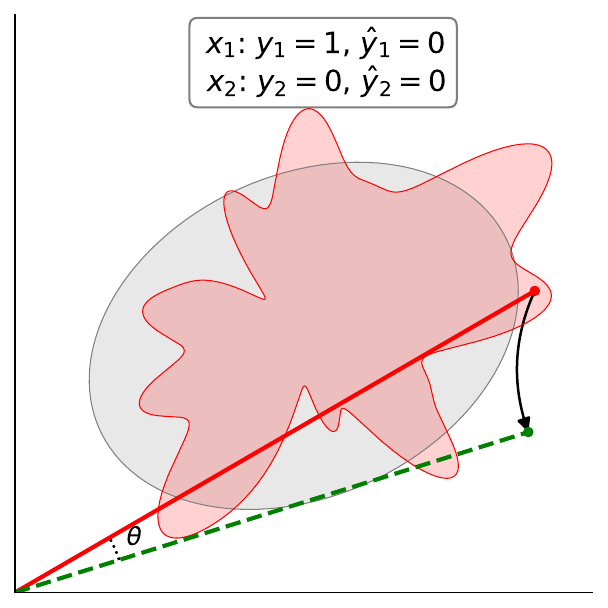}
    \centering \raisebox{0pt}{\textbf{(f)}}
\end{minipage}\hfill
\begin{minipage}{0.2375\textwidth}
    \includegraphics[width=\textwidth]{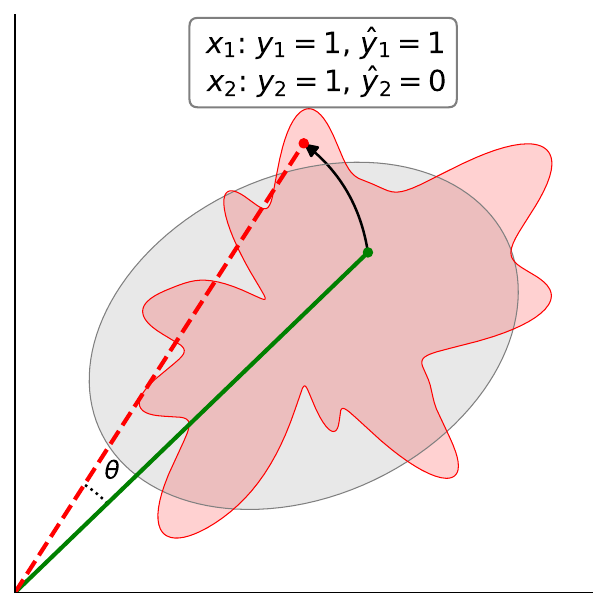}
    \centering \raisebox{0pt}{\textbf{(g)}}
\end{minipage}\hfill
\begin{minipage}{0.2375\textwidth}
    \includegraphics[width=\textwidth]{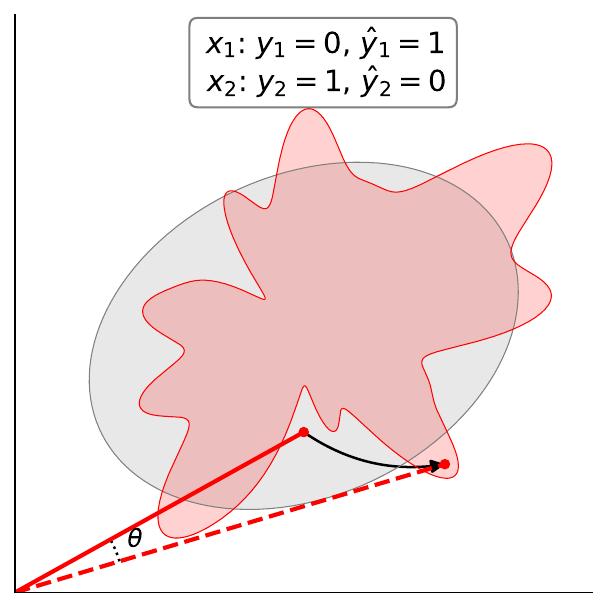}
    \centering \raisebox{0pt}{\textbf{(h)}}
\end{minipage}
\caption{Illustration of embedding-based defect prediction under file evolution. The red irregular region denotes the true defective boundary, while the gray elliptical region represents the estimated decision boundary. Each subfigure depicts an evolution pair $(x_1, x_2)$ with ground-truth labels $(y_1, y_2)$ and predicted labels $(\hat{y}_1, \hat{y}_2)$. (a--b) Label-persisted cases (B00, D11). (c--d) Label-changing cases (B10, D01). (e--f) \emph{Wrong in being right} phenomenon. (g) Benign-class bias. (h) Weak models misclassifying both versions. Other cases are omitted due to symmetry or low frequency.}
\label{fig:multi_fig}
\end{figure*}

\paragraph{High Unchanged-F1 and Low changed-F1} \autoref{fig:multi_fig}\hyperref[fig:multi_fig]{~(a)} and \autoref{fig:multi_fig}\hyperref[fig:multi_fig]{~(b)} show how this geometry favors label-persisted evolution pairs (B00 and D11). 
Because embeddings change little across versions, predictions are naturally reused, inflating unchanged-set performance (as seen in Section~\ref{sec:RQ1}). In contrast, \autoref{fig:multi_fig}\hyperref[fig:multi_fig]{~(c)} and \autoref{fig:multi_fig}\hyperref[fig:multi_fig]{~(d)} show how the same logic degrades performance on label-changing subsets (B10 and D01).

\paragraph{Wrong in Being Right} \autoref{fig:multi_fig}\hyperref[fig:multi_fig]{~(e)} and \autoref{fig:multi_fig}\hyperref[fig:multi_fig]{~(f)} reveal a more subtle but critical mode, which we call \emph{wrong in being right}. When a classifier learns an incorrect decision boundary around $x_1$, this error can be subtly propagated to $x_2$ due to high embedding similarity. \textbf{If the true label flips across versions, the propagated error is evaluated as correct}. This issue is evident in \autoref{tab:lucene_changed_correctness}. 

\begin{table}[!htbp]
\centering
\renewcommand{\arraystretch}{1.5} 
\setlength{\tabcolsep}{12pt} 
\caption{Train–test correctness matrices for the weak linear SVM model on B10 and D01 subsets of common files. Changed-F1 scores are: 0.551 (train), 0.381 (test), and 0.0 (test without propagated errors).}
\begin{tabular}{|c|c|c|c|}
\hline
Subset & Train/Test & $\hat{y}_2 = y_2$ & $\hat{y}_2 \neq y_2$ \\
\hline
\multirow{2}{*}{\textbf{B10 (1→0)}} & $\hat{y}_1 = y_1$ & 0  & 22 \\
\cline{2-4}
 & $\hat{y}_1 \neq y_1$ & \textbf{44} & 0 \\
\hline
\multirow{2}{*}{\textbf{D01 (0→1)}} & $\hat{y}_1 = y_1$ & 0  & 16 \\
\cline{2-4}
 & $\hat{y}_1 \neq y_1$ & \textbf{9} & 0 \\
\hline
\end{tabular}
\label{tab:lucene_changed_correctness}
\end{table}

\subsection{Why Change-Aware Formulation is Necessary}
\label{sec:justification_change_aware}

Consider the traditional SDP formulation that uses only the current file $X_{t+1}$ for prediction. 
\autoref{thm:label-persistence} shows that the Bayes-optimal predictor collapses to the naive label-persistent classifier (\autoref{fig:naive_classifier}). Furthermore, a change-aware predictor that conditions on the differences $D$ is statistically well-posed (\autoref{thm:necessity-of-change}). Check Appendix~\ref{appendix:bayes_proof} and Appendix~\ref{appendix:necessity_change} for the proof.
\begin{theorem}[Bayes-Optimality of Naive Label-Persistent Classifier]
\label{thm:label-persistence}
Let $f^\star:\mathcal U \to \{0,1\}$ be the Bayes-optimal classifier:
\[
f^\star(x) = \arg\max_{y \in \{0,1\}} \mathbb P(Y_{t+1}=y \mid X_{t+1}=x).
\]
Suppose label changes occur with small probability: $\mathbb P(Y_{t+1} \neq Y_t) = \varepsilon \in (0, \tfrac12)$ \text{(Empirically shown in Section \ref{sec:RQ1})}.

Then, with probability at least $1-\varepsilon$,
\[
f^\star(X_{t+1}) = Y_t.
\]
\end{theorem}
\begin{theorem}[Necessity of Conditioning on Change]
\label{thm:necessity-of-change}
No predictor based solely on $X_{t+1}$ can consistently estimate defect
transition probabilities.
\end{theorem}
When distinct pairs $(X_t, D)$ lead to the same $X_{t+1}$, the change $D$ cannot be recovered from $X_{t+1}$ alone. Conditioning on $X_{t+1}$ therefore merges cases with and without changes. As a result, $\mathbb{P}(Y_{t+1} \mid X_{t+1})$ cannot capture the transition behavior driven by changes, and predictors that ignore $D$ cannot model defect transitions. Conditioning on the change $D$ separates \emph{what the file is} from \emph{how it got there}, allowing the model to focus on the causal event that drives label transitions.

\subsection{Comparative Analysis and Performance}
\label{sec:final_comparison}
We complement Section~\ref{sec:RQ3} by comparing all methods and baselines in \autoref{tab:compare-results}, summarizing the insights aimed to establish through the research questions proposed in Section~\ref{sec:research_questions}. The multi-agent debate framework significantly improves accuracy on the most challenging status-transition subset, namely D01. A useful guiding principle in the context of change-aware SDP is that
\begin{center}
\emph{an effective method, bypassing the illusion of accuracy, achieves high F1 in conjunction with high HMB and HMD.}
\end{center}

\begin{table*}[!htbp]
\centering
\scriptsize
\setlength{\tabcolsep}{6pt}
\renewcommand{\arraystretch}{1.1}
\caption{Performance comparison of methods and baselines on Lucene project. We use DeepSeek-Reasoning and Qwen2.5-Coder-32B for M0 and M5 respectively.
$^{*}$ For Random Forest, metrics are computed after excluding propagated errors explained in Section~\ref{sec:Rethinking_SDP}.
}
\label{tab:compare-results}

\resizebox{\textwidth}{!}{%
\begin{tabular}{l
cccc
cc
ccc}
\toprule
& \multicolumn{4}{c}{\textbf{Subset Accuracies}}
& \multicolumn{2}{c}{\textbf{Change-aware}}
& \multicolumn{3}{c}{\textbf{F1 Scores}} \\
\cmidrule(lr){2-5}
\cmidrule(lr){6-7}
\cmidrule(lr){8-10}
\textbf{Method}
& \textbf{B00} & \textbf{D01} & \textbf{D11} & \textbf{B10}& \textbf{HMB} & \textbf{HMD}
& \textbf{Changed} & \textbf{Unchanged} & \textbf{Total} \\
\midrule
Naive All Benign Prediction
& 1.00 & \cellcolor{red!20}0.00 & 0.00 & 1.00
& \cellcolor{red!30}0.00 & \cellcolor{red!30}0.00
& 0.42 & 0.49 & 0.48 \\
Naive Label-persistent Classifier
& 1.00  & \cellcolor{red!20}0.00 & 1.00 & 0.00
& \cellcolor{red!30}0.00  & \cellcolor{red!30}0.00
& 0.0  & 1.0  & 0.73  \\
\midrule
Random Forest*
& 1.00  & \cellcolor{red!20}0.00  & 0.68  & {0.04} 
& \cellcolor{red!30}0.00  & \cellcolor{red!30}{0.07} 
&{0.00}  & 0.90  & 0.70  \\
M0 (Single-version Classification)
& 0.99  & \cellcolor{red!20}0.00  & 0.37  & 0.86 
& \cellcolor{red!30}0.00  & \cellcolor{green!25}0.51 
& 0.49  & 0.49  & 0.50  \\
M5 (Diff-only Reasoning)
& 0.99  & \cellcolor{red!20}0.00  & 0.63 & 0.62 
& \cellcolor{red!30}0.00  & \cellcolor{green!30}0.62
& 0.31  & 0.85  & 0.70  \\
Multi-Agent Debate
& 0.71 & \cellcolor{green!20}0.52 & 0.80 & 0.18
& \cellcolor{green!35}0.60 & \cellcolor{green!20}0.30
& 0.27 & 0.52 & 0.51  \\
\bottomrule
\end{tabular}
}
\end{table*}

\subsection{Impact of Debate Rounds and Comprehensive Performance of Multi-Agent Debate Across PROMISE Datasets}
\autoref{tab:mad_overall_performance} shows that a single debate round is generally sufficient and exhibits dominant performance across PROMISE datasets. One-round debate consistently achieves stronger results on the B00 subset while maintaining often stable and competitive HMB and HMD scores, leading to a clear and significant advantage in Total F1. Additional rounds rarely close this gap and instead introduce variability without consistent overall gains.

\begin{table*}[t]
\centering
\caption{Impact of Debate Rounds on the Overall Performance of the Multi-Agent Debate Framework Across PROMISE Datasets. Best results per dataset and metric are shown in bold. $^{*}$ Zero accuracy on the B10 subset for Groovy is due to the extremely small size of only three instances.}
\label{tab:mad_overall_performance}
\resizebox{0.90\textwidth}{!}{%
\begin{tabular}{c c c c c c c c c c c}
\toprule
{Dataset} 
& {Debate} 
& \multicolumn{4}{c}{{Subset Accuracies}} 
& \multicolumn{2}{c}{{Change-aware}} 
& \multicolumn{3}{c}{{F1 Scores}} \\
\cmidrule(lr){3-6}
\cmidrule(lr){7-8}
\cmidrule(lr){9-11}
& {Rounds}
& {B00} 
& {D01} 
& {D11} 
& {B10} 
& {HMB} 
& {HMD} 
& {Changed} 
& {Unchanged} 
& {Total} \\
\midrule

\multirow{3}{*}{Lucene}
 & 1 & \textbf{0.72} & \underline{0.52} & \underline{0.80} & \underline{0.18} & \textbf{0.60} & \underline{0.30} & \underline{0.28} & \textbf{0.53} & \textbf{0.52} \\
 & 2 & 0.58 & 0.48 & \textbf{0.96} & 0.15 & 0.52 & 0.27 & 0.24 & 0.46 & 0.45 \\
& 3 & \underline{0.66} & \textbf{0.54} & 0.79 & \textbf{0.25} & \underline{0.59} & \textbf{0.38} & \textbf{0.33} & \underline{0.48} & \underline{0.49} \\
\midrule

\multirow{3}{*}{Groovy}
 & 1 & \textbf{0.56} & \textbf{0.50} & \underline{0.75} & 0.00* & \textbf{0.53} & 0.00* & \underline{0.28} & \textbf{0.43} & \textbf{0.46} \\
 & 2 & \underline{0.45} & \underline{0.40} & \textbf{0.88} & \textbf{0.33} & \underline{0.42} & \textbf{0.48} & \textbf{0.35} & \underline{0.39} & \underline{0.40} \\
 & 3 & {0.45} & 0.30 & 0.50 & \underline{0.33} & {0.36} & \underline{0.40} & 0.29 & 0.35 & 0.36 \\
\midrule

\multirow{3}{*}{Derby}
 & 1 & \textbf{0.67} & 0.32 & {0.77} & \textbf{0.23} & 0.43 & \textbf{0.35} & \textbf{0.21} & \textbf{0.61} & \textbf{0.50} \\
 & 2 & \underline{0.53} & \textbf{0.74} & \underline{0.83} & \underline{0.16} & \textbf{0.62} & \underline{0.27} & \underline{0.21} & \underline{0.53} & \underline{0.45} \\
 & 3 & 0.47 & \underline{0.50} & \textbf{0.95} & 0.15 & \underline{0.48} & {0.27} & 0.18 & 0.42 & 0.39 \\
\midrule

\multirow{3}{*}{HBase}
 & 1 & \textbf{0.75} & 0.29 & 0.72 & \textbf{0.29} & {0.42} & \textbf{0.42} & \textbf{0.29} & \textbf{0.57} & \textbf{0.51} \\
 & 2 & \underline{0.62} & \underline{0.42} & \underline{0.85} & 0.13 & \textbf{0.50} & {0.23} & 0.23 & \underline{0.51} & \underline{0.46} \\
 & 3 & 0.50 & \textbf{0.51} & \textbf{0.91} & \underline{0.18} & \underline{0.50} & \underline{0.30} & \underline{0.29} & {0.44} & 0.43 \\
\midrule
\multirow{3}{*}{JRuby}
 & 1 & \textbf{0.61} & 0.40 & \textbf{0.93} & \textbf{0.29} & 0.48 & \textbf{0.44} & \textbf{0.24} & \textbf{0.49} & \textbf{0.41} \\
 & 2 & 0.51 & \textbf{0.67} & 0.86 & \underline{0.12} & \textbf{0.58} & \underline{0.21} & \underline{0.15} & 0.43 & 0.34 \\
 & 3 & \underline{0.55} & \underline{0.60} & \underline{0.93} & 0.11 & \underline{0.57} & 0.20 & 0.13 & \underline{0.46} & \underline{0.36} \\
\midrule

\multirow{3}{*}{Camel}
& 1 & \textbf{0.76} & 0.30 & \underline{0.90} & \textbf{0.35} & 0.43 & \textbf{0.50} & \underline{0.32} & \textbf{0.47} & \textbf{0.48} \\
& 2 & 0.67 & \underline{0.48} & 0.85 & \underline{0.29} & \textbf{0.56} & \underline{0.43} & \textbf{0.35} & \underline{0.43} & \underline{0.44} \\
& 3 & \underline{0.60} & \textbf{0.49} & \textbf{0.95} & {0.18} & \underline{}{0.54} & 0.30 & {0.27} & {0.40} & {0.41} \\

\bottomrule
\end{tabular}
}
\end{table*}

\subsection{Time Complexity Analysis}
Figure~\ref{fig:time_complexity} shows the effects of parameters and code diff size on runtime complexity of multi-agent debate framework.

\begin{figure*}[!htbp]
    \centering

    \begin{minipage}[t]{0.48\textwidth}
        \centering
        \includegraphics[height=8.2cm]{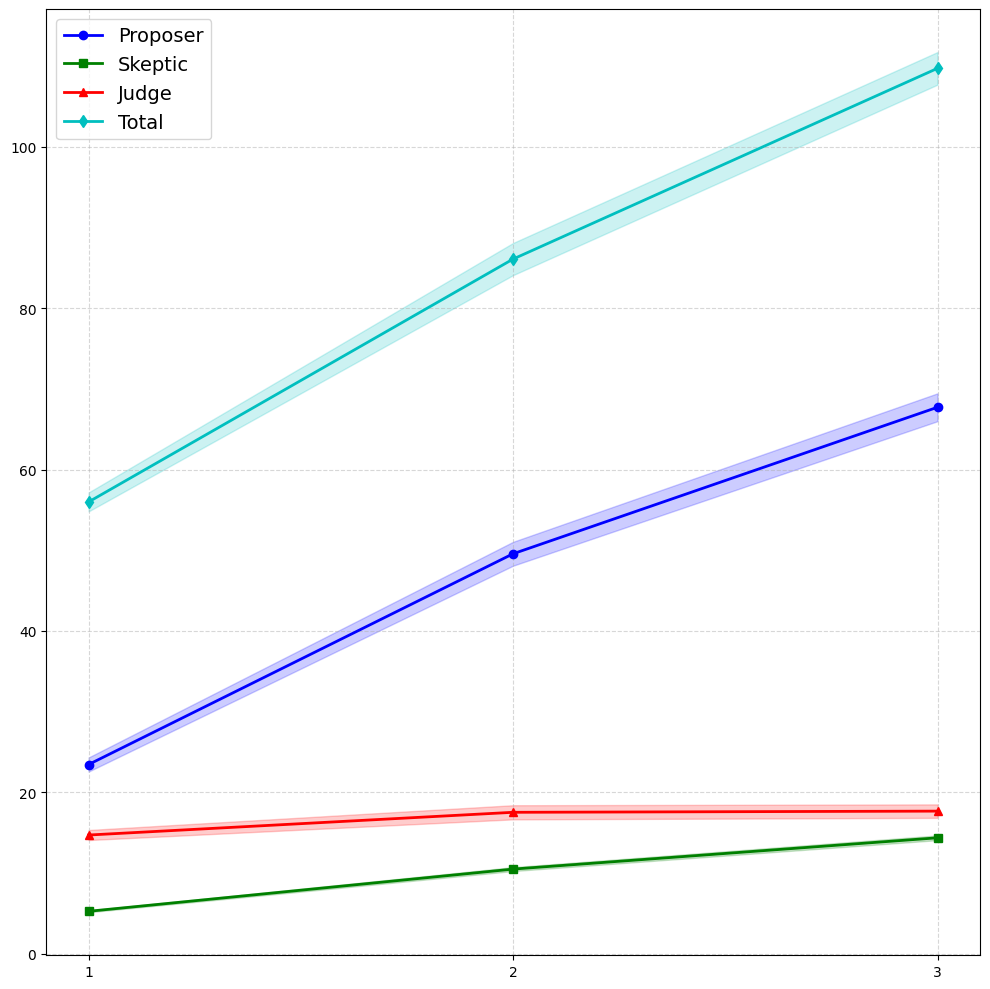}
        \vspace{0.2cm}
        
        \small (a) vs Number of Debates
    \end{minipage}
    \hfill
    \begin{minipage}[t]{0.48\textwidth}
        \centering
        \includegraphics[height=8.2cm]{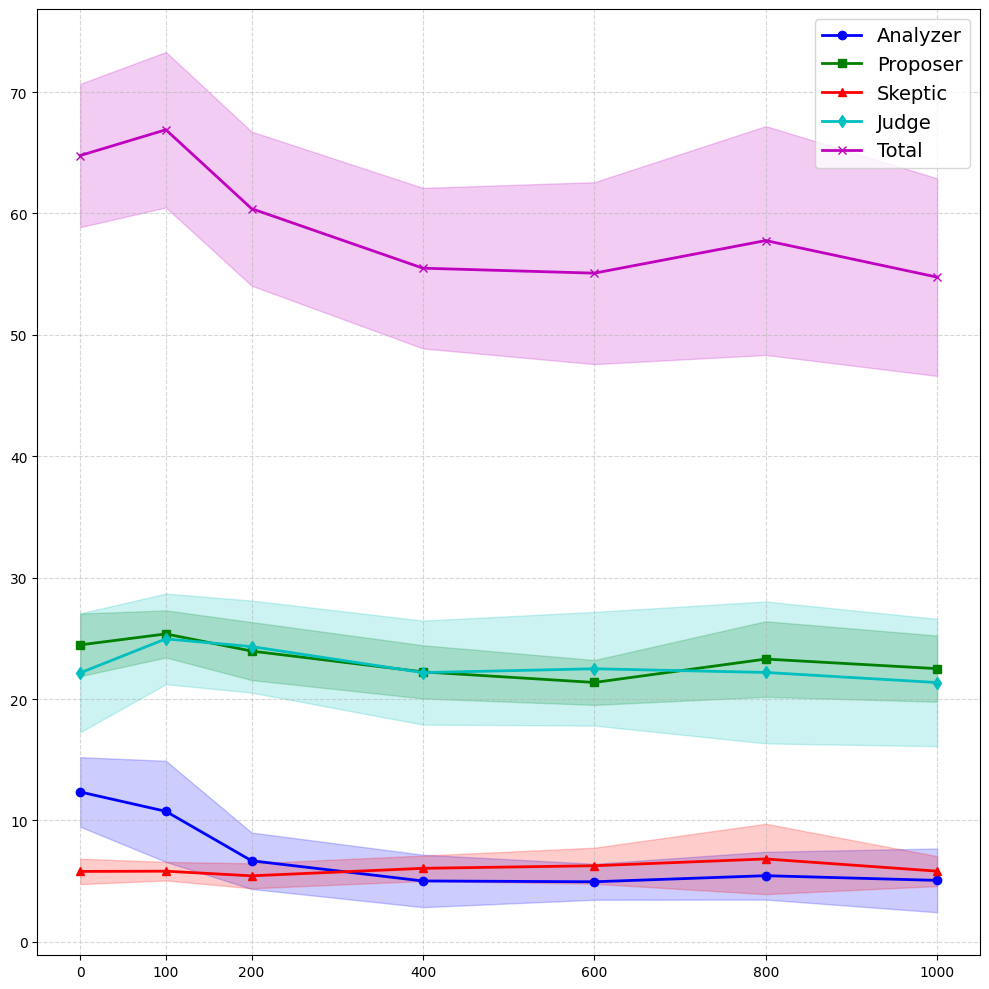}
        \vspace{0.2cm}
        
        \small (b) vs Max Lines
    \end{minipage}

    \vspace{0.35cm}

    \begin{minipage}[t]{0.95\textwidth}
        \centering
        \includegraphics[width=\textwidth]{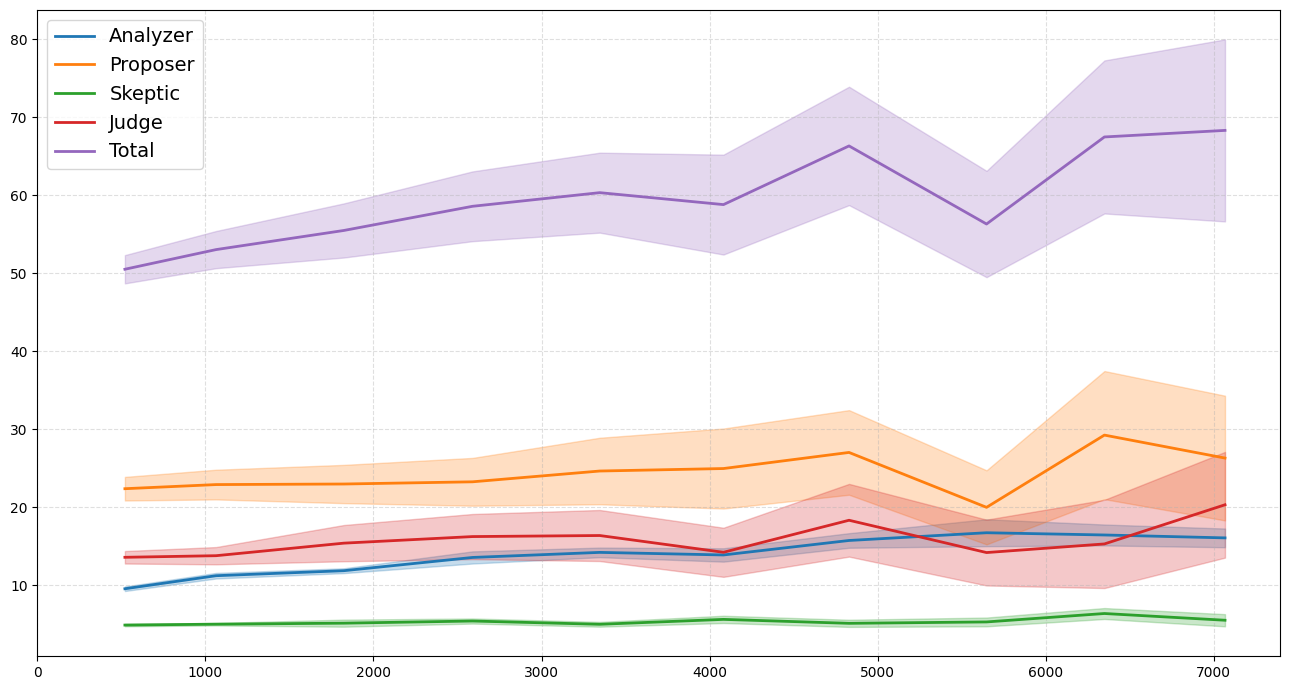}
        \vspace{0.2cm}
        
        \small (c) vs Source Code Change Size
    \end{minipage}

    \caption{
        Runtime complexity analysis of the multi-agent debate framework (average time in seconds). Shaded regions denote the 95\% confidence interval
    }
    \label{fig:time_complexity}
\end{figure*}

\subsection{Limitations}
The proposed change-aware formulation relies on access to the defect status of a file in the previous version. While this assumption naturally follows from the change-based setting and versioned repositories, it restricts applicability in scenarios where historical labels are missing or cannot be reliably annotated. 
Our problem formulation and evaluation is limited to common files, excluding newly added ones. The file-matching strategy assumes a one-to-one correspondence across versions and does not account for many-to-one or one-to-many evolution patterns, such as file merges and splits, which are therefore outside the scope of this work.
Additionally, the context expansion algorithm in Section~\ref{sec:subset_motivation_retrieval} ignores the importance of different context elements and does not distinguish more informative methods from less relevant ones.
Although this is still more informative than using the change in isolation, all retrieved context is treated uniformly, which can include redundant code or lead to the truncation of more informative parts under the \texttt{max\_lines} limit.

\section{Threats to Validity}
\label{sec:validity}
We discuss potential threats that may affect the credibility and generalizability of our results, grouped into internal, external, and construct validity.
\subsection{Internal Validity}
Internal validity concerns whether the observed effects stem from our methods rather than uncontrolled factors.  Our similarity-based file-matching may misalign files during major refactorings, influencing the evolution subsets (B00, D01, B10, D11). LLM baselines are sensitive to prompt design and hyperparameters; despite using standardized templates, minor variations may affect reasoning. The multi-agent debate framework mitigates hallucination but may still converge to incorrect consensus when both agents share similar biases.
\subsection{External Validity}
External validity relates to the generalizability of our findings. Experiments on PROMISE Java projects may not fully represent modern industrial settings such as microservices, polyglot systems, or rapid-release environments. Our file-level focus may differ from real-world defect management at commit, function, or line-level. Moreover, reliance on large language models may limit reproducibility in organizations with cost, privacy, or licensing constraints.
\subsection{Construct Validity}
Construct validity examines whether our design and metrics capture the intended concepts. Change-aware defect prediction depends on accurate status-transition labels; noise in PROMISE bug labels may affect subsets (B10 and D01). Changed-F1, 
Finally, the context expansion algorithm focuses on local caller-callee relations, which may overlook broader architectural dependencies in highly modular systems.
\section{Conclusion and Future Works}
\label{sec:conclusion}
Traditional single-file formulation of SDP is fundamentally misaligned with the realities of software evolution, and its reported success is largely driven by an \textit{"illusion of accuracy"} caused by label-persistence bias. By conditioning predictions on code changes, we demonstrated that a change-aware formulation is not only more informative but necessary to avoid pitfalls such as deceptively-high F1 score and the \textit{"wrong in being right"} phenomenon. Our analysis further revealed that a meaningful evaluation must itself be change-aware, leading us to metrics such as changed-F1 and unchanged-F1, subset-wise accuracies (B00,D01,D11,B10), and the corresponding harmonic means (HMB and HMD). 

We found that LLMs are less deceptive than traditional embedding-based approaches, which consistently fail on defect-introducing and defect-fixing cases despite high total F1. Nonetheless, single-agent LLMs still struggle with newly introduced defects (D01), a limitation that is effectively solved by our multi-agent debate framework, yielding more balanced and reliable performance across all evolution subsets. Together, these findings argue for a paradigm shift in SDP toward change-aware modeling and evaluation, and suggest that multi-agent debate framework is a promising direction for reliable change-aware software defect prediction.

Future work includes relaxing the reliance on previous defect labels to support change-aware prediction when historical annotations are missing or unreliable. Another important direction is extending the one-to-one file-evolution assumption to handle many-to-one and one-to-many changes, such as file merges and splits.
Additionally, the multi-agent debate framework, while unaffected by the illusion of accuracy, still falls far short of practical performance, leaving considerable room for improvement despite its high complexity. It is valuable to explore alternative approaches with less computational cost and improved performance.
Possible directions include better context expansion strategies to prioritize the most informative code elements and critical features. Another avenue is designing a multi-agent framework with a possibly simpler yet more efficient information flow. It is also helpful to create benchmarks that capture a wider variety of code-change patterns for each file. Perhaps, with a sufficiently large and rich dataset on code changes, simpler models may move beyond illusion and progress toward generating genuine insight.





\bibliographystyle{ieeetran}

\bibliography{cas-refs}

\begin{thebibliography}{100}
\providecommand{\url}[1]{#1}
\csname url@samestyle\endcsname
\providecommand{\newblock}{\relax}
\providecommand{\bibinfo}[2]{#2}
\providecommand{\BIBentrySTDinterwordspacing}{\spaceskip=0pt\relax}
\providecommand{\BIBentryALTinterwordstretchfactor}{4}
\providecommand{\BIBentryALTinterwordspacing}{\spaceskip=\fontdimen2\font plus
\BIBentryALTinterwordstretchfactor\fontdimen3\font minus \fontdimen4\font\relax}
\providecommand{\BIBforeignlanguage}[2]{{%
\expandafter\ifx\csname l@#1\endcsname\relax
\typeout{** WARNING: IEEEtran.bst: No hyphenation pattern has been}%
\typeout{** loaded for the language `#1'. Using the pattern for}%
\typeout{** the default language instead.}%
\else
\language=\csname l@#1\endcsname
\fi
#2}}
\providecommand{\BIBdecl}{\relax}
\BIBdecl

\bibitem{Li2018}
Z.~Li, P.~He, J.~Zhu, and M.~R. Lyu, ``Software defect prediction via convolutional neural network,'' \emph{IEEE Transactions on Reliability}, vol.~67, no.~3, pp. 885--896, 2018.

\bibitem{Hall2012}
T.~Hall, S.~Beecham, D.~Bowes, D.~Gray, and S.~Counsell, ``A systematic literature review on fault prediction performance in software engineering,'' \emph{IEEE Transactions on Software Engineering}, vol.~38, no.~6, pp. 1276--1304, 2012.

\bibitem{Kitchenham2017}
B.~A. Kitchenham, L.~Madeyski, D.~Budgen, and J.~Keung, ``Robust statistical methods for empirical software engineering,'' \emph{Empirical Software Engineering}, vol.~22, no.~2, pp. 579--630, 2017.

\bibitem{Rahman2013}
F.~Rahman and P.~Devanbu, ``How, and why, process metrics are better,'' in \emph{Proceedings of the 35th International Conference on Software Engineering (ICSE)}.\hskip 1em plus 0.5em minus 0.4em\relax IEEE, 2013, pp. 432--441.

\bibitem{Tantithamthavorn2018}
C.~Tantithamthavorn, A.~E. Hassan, and K.~Matsumoto, ``The impact of automated parameter optimization on defect prediction models,'' \emph{IEEE Transactions on Software Engineering}, vol.~45, no.~7, pp. 683--711, 2018.

\bibitem{Hosseini2017}
S.~Hosseini, B.~Turhan, and S.~Gunarathna, ``A systematic literature review and meta-analysis on cross project defect prediction,'' \emph{IEEE Transactions on Software Engineering}, vol.~45, no.~2, pp. 111--144, 2017.

\bibitem{promise_awsm}
{Promise Defects Projects}, ``Promise software defect prediction dataset,'' \url{https://github.com/awsm-research/line-level-defect-prediction}.

\bibitem{nasa_klainfo}
{NASA Metrics Data Program}, ``Nasa mdp software defect datasets,'' \url{https://github.com/klainfo/NASADefectDataset}.

\bibitem{aeeem_kaggle}
S.~K.~G. RAHUL~KASAUDHAN, ``Aeeem dataset for software defect prediction,'' \url{https://www.kaggle.com/datasets/rahulkasaudhan/aeeem-dataset}.

\bibitem{nevendra2022survey}
M.~Nevendra and P.~Singh, ``A survey of software defect prediction based on deep learning,'' \emph{Archives of Computational Methods in Engineering}, vol.~29, no.~7, pp. 5723--5748, 2022.

\bibitem{zhou2024research}
X.~Zhou, X.~Wang, J.~Gan, Y.~Liu, and W.~Wang, ``Research on software defect detection model based on neural network,'' in \emph{2024 International Conference on Information Technology, Comunication Ecosystem and Management (ITCEM)}.\hskip 1em plus 0.5em minus 0.4em\relax IEEE, 2024, pp. 47--54.

\bibitem{liu2025improved}
Z.~Liu and X.~Zhou, ``An improved transformer-based model for software defect prediction,'' \emph{Deep Learning and Pattern Recognition}, vol.~1, no.~1, 2025.

\bibitem{Majd2021}
A.~e.~a. Majd, ``Sldeep: Statement-level deep learning for software defect prediction,'' \emph{Information and Software Technology}, 2021.

\bibitem{Zhou2019}
Y.~Zhou, S.~Wang, and D.~Lo, ``Devign: Effective vulnerability identification by learning comprehensive program semantics via graph neural networks,'' in \emph{Advances in Neural Information Processing Systems (NeurIPS)}, 2019, pp. 4910--4921.

\bibitem{Hoang2020}
T.~Hoang, H.~K. Dam, Y.~Kamei, D.~Lo, and N.~Ubayashi, ``Deepjit: an end-to-end deep learning framework for just-in-time defect prediction,'' in \emph{2019 IEEE/ACM 16th International Conference on Mining Software Repositories (MSR)}.\hskip 1em plus 0.5em minus 0.4em\relax IEEE, 2019, pp. 34--45.

\bibitem{Dam2019}
H.~K. Dam, T.~Tran, J.~Grundy, and A.~Ghose, ``Deepsoft: A vision for a deep model of software,'' in \emph{Proceedings of the 2016 24th ACM SIGSOFT international symposium on foundations of software engineering}, 2016, pp. 944--947.

\bibitem{Adediran2024}
\BIBentryALTinterwordspacing
E.~Adediran, ``Language-proficient transformer for software defect prediction,'' \emph{SSRN Electronic Journal}, 2024. [Online]. Available: \url{https://papers.ssrn.com/sol3/papers.cfm?abstract_id=5269565}
\BIBentrySTDinterwordspacing

\bibitem{Feng2020}
Z.~e.~a. Feng, ``Codebert: A pre-trained model for programming and natural languages,'' \emph{EMNLP}, 2020.

\bibitem{Wang2023}
Y.~e.~a. Wang, ``Codet5+: Open code large language models for code understanding and generation,'' \emph{arXiv preprint arXiv:2305.07922}, 2023.

\bibitem{yang2019xlnet}
Z.~Yang, Z.~Dai, Y.~Yang, J.~Carbonell, R.~R. Salakhutdinov, and Q.~V. Le, ``Xlnet: Generalized autoregressive pretraining for language understanding,'' \emph{Advances in neural information processing systems}, vol.~32, 2019.

\bibitem{lozhkov2024starcoder}
A.~Lozhkov, R.~Li, L.~B. Allal, F.~Cassano, J.~Lamy-Poirier, N.~Tazi, A.~Tang, D.~Pykhtar, J.~Liu, Y.~Wei \emph{et~al.}, ``Starcoder 2 and the stack v2: The next generation,'' \emph{arXiv preprint arXiv:2402.19173}, 2024.

\bibitem{wei2022chain}
J.~Wei, X.~Wang, D.~Schuurmans, M.~Bosma, F.~Xia, E.~Chi, Q.~V. Le, D.~Zhou \emph{et~al.}, ``Chain-of-thought prompting elicits reasoning in large language models,'' \emph{Advances in neural information processing systems}, vol.~35, pp. 24\,824--24\,837, 2022.

\bibitem{nam2024using}
D.~Nam, A.~Macvean, V.~Hellendoorn, B.~Vasilescu, and B.~Myers, ``Using an llm to help with code understanding,'' in \emph{Proceedings of the IEEE/ACM 46th International Conference on Software Engineering}, 2024, pp. 1--13.

\bibitem{clinton2025proactive}
\BIBentryALTinterwordspacing
D.~Clinton, ``Proactive software reliability: Deep learning models for early fault prediction and adaptive risk mitigation,'' \emph{ResearchGate Preprint}, 2025, this work models temporal dependencies in evolving code bases for proactive defect prediction. [Online]. Available: \url{https://www.researchgate.net/publication/396460466_Proactive_Software_Reliability_Deep_Learning_Models_for_Early_Fault_Prediction_and_Adaptive_Risk_Mitigation}
\BIBentrySTDinterwordspacing

\bibitem{deeptha2024robust}
\BIBentryALTinterwordspacing
R.~Deeptha and K.~Krishnamoorthi, ``Designing a robust software bug prediction model using enhanced learning principles with artificial intelligence assistance,'' in \emph{IEEE Conference on Innovative Data Science}, 2024, focuses on incorporating temporal features of software evolution in defect prediction. [Online]. Available: \url{https://ieeexplore.ieee.org/abstract/document/10910397}
\BIBentrySTDinterwordspacing

\bibitem{fakih2025llm4cve}
\BIBentryALTinterwordspacing
M.~Fakih, R.~Dharmaji, H.~Bouzidi, and G.~Araya, ``Llm4cve: Enabling iterative automated vulnerability repair with large language models,'' \emph{arXiv preprint arXiv:2501.03446}, 2025, proposes an iterative repair pipeline where LLMs detect, explain, and patch vulnerabilities, linking prediction with prevention. [Online]. Available: \url{https://arxiv.org/abs/2501.03446}
\BIBentrySTDinterwordspacing

\bibitem{yu2025graphrag}
\BIBentryALTinterwordspacing
H.~Yu and B.~Scanlon, ``Engineering critical analysis software services: A graph-rag and self-learning large language model agent services approach,'' in \emph{IEEE Conference on Service Engineering}, 2025, introduces a self-learning LLM framework combining reasoning and repair for continuous software reliability improvement. [Online]. Available: \url{https://ieeexplore.ieee.org/abstract/document/11126127}
\BIBentrySTDinterwordspacing

\bibitem{zeng2021deep}
Z.~Zeng, Y.~Zhang, H.~Zhang, and L.~Zhang, ``Deep just-in-time defect prediction: how far are we?'' in \emph{Proceedings of the 30th ACM SIGSOFT international symposium on software testing and analysis}, 2021, pp. 427--438.

\bibitem{li2025swe}
H.~Li, Y.~Shi, S.~Lin, X.~Gu, H.~Lian, X.~Wang, Y.~Jia, T.~Huang, and Q.~Wang, ``Swe-debate: Competitive multi-agent debate for software issue resolution,'' \emph{arXiv preprint arXiv:2507.23348}, 2025.

\bibitem{tillmann2025literature}
A.~Tillmann, ``Literature review of multi-agent debate for problem-solving,'' \emph{arXiv preprint arXiv:2506.00066}, 2025.

\bibitem{vamosi2025crawdad}
F.~G. Vamosi and N.~D. Forkert, ``Crawdad: Causal reasoning augmentation with dual-agent debate,'' \emph{arXiv preprint arXiv:2511.22854}, 2025.

\bibitem{lessmann2008benchmarking}
S.~Lessmann, B.~Baesens, C.~Mues, and S.~Pietsch, ``Benchmarking classification models for software defect prediction: A proposed framework and novel findings,'' \emph{IEEE Transactions on Software Engineering}, vol.~34, no.~4, pp. 485--496, 2008.

\bibitem{hall2012systematic}
T.~Hall, S.~Beecham, D.~Bowes, D.~Gray, and S.~Counsell, ``A systematic literature review on fault prediction performance in software engineering,'' \emph{IEEE Transactions on Software Engineering}, vol.~38, no.~6, pp. 1276--1304, 2012.

\bibitem{gray2010software}
D.~Gray, D.~Bowes, N.~Davey, Y.~Sun, and B.~Christianson, ``Software defect prediction using static code metrics underestimates defect-proneness,'' in \emph{The 2010 International Joint Conference on Neural Networks (IJCNN)}.\hskip 1em plus 0.5em minus 0.4em\relax IEEE, 2010, pp. 1--7.

\bibitem{elish2008predicting}
K.~O. Elish and M.~O. Elish, ``Predicting defect-prone software modules using support vector machines,'' \emph{Journal of Systems and Software}, vol.~81, no.~5, pp. 649--660, 2008.

\bibitem{bennin2018mahakil}
K.~E. Bennin, J.~Keung, P.~Phannachitta, A.~Monden, and S.~Mensah, ``Mahakil: Diversity based oversampling approach to alleviate the class imbalance issue in software defect prediction,'' \emph{IEEE Transactions on Software Engineering}, vol.~44, no.~6, pp. 534--550, 2018.

\bibitem{liu2024cfg2at}
H.~Liu, Z.~Li, H.~Zhang, X.-Y. Jing, and J.~Liu, ``Cfg2at: Control flow graph and graph attention network-based software defect prediction,'' \emph{IEEE Transactions on Reliability}, 2024.

\bibitem{o2015introduction}
K.~O'shea and R.~Nash, ``An introduction to convolutional neural networks,'' \emph{arXiv preprint arXiv:1511.08458}, 2015.

\bibitem{sherstinsky2020fundamentals}
A.~Sherstinsky, ``Fundamentals of recurrent neural network (rnn) and long short-term memory (lstm) network,'' \emph{Physica D: Nonlinear Phenomena}, vol. 404, p. 132306, 2020.

\bibitem{wang2018deep}
S.~Wang, T.~Liu, J.~Nam, and L.~Tan, ``Deep semantic feature learning for software defect prediction,'' \emph{IEEE Transactions on Software Engineering}, vol.~46, no.~12, pp. 1267--1293, 2018.

\bibitem{majd2020sldeep}
A.~Majd, M.~Vahidi-Asl, A.~Khalilian, P.~Poorsarvi-Tehrani, and H.~Haghighi, ``Sldeep: Statement-level software defect prediction using deep-learning model on static code features,'' \emph{Expert Systems with Applications}, vol. 147, p. 113156, 2020.

\bibitem{pornprasit2023deeplinedp}
C.~Pornprasit and C.~Tantithamthavorn, ``Deeplinedp: Towards a deep learning approach for line-level defect prediction,'' \emph{IEEE Transactions on Software Engineering}, vol.~49, no.~1, pp. 84--98, 2023.

\bibitem{liu2023semantic}
J.~Liu, J.~Ai, M.~Lu, J.~Wang, and H.~Shi, ``Semantic feature learning for software defect prediction from source code and external knowledge,'' \emph{Journal of Systems and Software}, vol. 206, p. 111753, 2023.

\bibitem{xu2020defect}
J.~Xu, F.~Wang, and J.~Ai, ``Defect prediction with semantics and context features of codes based on graph representation learning,'' \emph{IEEE Transactions on Reliability}, vol.~70, no.~2, pp. 613--625, 2020.

\bibitem{huo2018learning}
X.~Huo, Y.~Yang, M.~Li, and D.~Zhan, ``Learning semantic features for software defect prediction by code comments embedding,'' \emph{Proceedings of the IEEE International Conference on Data Mining}, pp. 1049--1054, 2018.

\bibitem{saeed2024cross}
M.~S. Saeed, ``Cross project software defect prediction using ensemble learning, a comprehensive review,'' \emph{International Journal of Computational and Innovative Sciences}, vol.~3, no.~2, pp. 34--42, 2024.

\bibitem{kamei2016just}
Y.~Kamei, T.~Fukushima, S.~McIntosh, K.~Yamashita, N.~Ubayashi, and A.~E. Hassan, ``Studying just-in-time defect prediction using cross-project models,'' \emph{Empirical Software Engineering}, vol.~21, pp. 2072--2106, 2016.

\bibitem{bai2022transfer}
J.~Bai, J.~Jia, and L.~F. Capretz, ``A three-stage transfer learning framework for multi-source cross-project software defect prediction,'' \emph{Information and Software Technology}, vol. 150, p. 106985, 2022.

\bibitem{tong2023array}
H.~Tong, W.~Lu, W.~Xing, and S.~Wang, ``Array: Adaptive triple feature-weighted transfer naive bayes for cross-project defect prediction,'' \emph{Journal of Systems and Software}, vol. 202, p. 111721, 2023.

\bibitem{sheng2020adversarial}
L.~Sheng, L.~Lu, and J.~Lin, ``An adversarial discriminative convolutional neural network for cross-project defect prediction,'' \emph{IEEE Access}, vol.~8, pp. 55\,241--55\,253, 2020.

\bibitem{xing2022cross}
Y.~Xing, X.~Qian, B.~Guan, B.~Yang, and Y.~Zhang, ``Cross-project defect prediction based on g-lstm model,'' \emph{Pattern Recognition Letters}, vol. 160, pp. 50--57, 2022.

\bibitem{gong2019tackling}
L.~Gong, S.~Jiang, and L.~Jiang, ``Tackling class imbalance problem in software defect prediction through cluster-based over-sampling with filtering,'' \emph{IEEE Access}, vol.~7, pp. 145\,725--145\,737, 2019.

\bibitem{pak2018smote}
C.~Pak, T.~T. Wang, and X.~Su, ``An empirical study on software defect prediction using over-sampling by smote,'' \emph{International Journal of Software Engineering and Knowledge Engineering}, vol.~28, no.~6, pp. 811--830, 2018.

\bibitem{feng2021coste}
S.~Feng, J.~Keung, X.~Yu, Y.~Xiao, K.~E. Bennin, M.~Kabir, and M.~Zhang, ``Coste: Complexity-based oversampling technique to alleviate the class imbalance problem in software defect prediction,'' \emph{Information and Software Technology}, vol. 129, p. 106432, 2021.

\bibitem{yedida2021value}
R.~Yedida and T.~Menzies, ``On the value of oversampling for deep learning in software defect prediction,'' \emph{IEEE Transactions on Software Engineering}, vol.~48, no.~8, pp. 3103--3116, 2021.

\bibitem{tong2018software}
H.~Tong, B.~Liu, and S.~Wang, ``Software defect prediction using stacked denoising autoencoders and two-stage ensemble learning,'' \emph{Information and Software Technology}, vol.~96, pp. 94--111, 2018.

\bibitem{zhao2021metaheuristic}
K.~Zhao, S.~Ying, N.~Zhang, and D.~Zhu, ``Software defect prediction based on enhanced metaheuristic feature selection optimization and a hybrid deep neural network,'' \emph{Journal of Systems and Software}, vol. 180, p. 111026, 2021.

\bibitem{tong2024master}
H.~Tong \emph{et~al.}, ``Master: Multi-source transfer weighted ensemble learning for multiple sources cross-project defect prediction,'' \emph{IEEE Transactions on Software Engineering}, vol.~50, no.~5, pp. 1281--1305, May 2024.

\bibitem{yu2021hierarchical}
H.~Yu, X.~Sun, Z.~Zhou, and G.~Fan, ``A novel software defect prediction method based on hierarchical neural network,'' \emph{Proceedings of the 2021 IEEE 45th Annual Computers, Software, and Applications Conference (COMPSAC)}, pp. 366--375, 2021.

\bibitem{zheng2021software}
W.~Zheng, L.~Tan, and C.~Liu, ``Software defect prediction method based on transformer model,'' in \emph{2021 IEEE International Conference on Artificial Intelligence and Computer Applications (ICAICA)}.\hskip 1em plus 0.5em minus 0.4em\relax IEEE, 2021, pp. 670--674.

\bibitem{pandey2025deep}
S.~K. Pandey, A.~Haldar, and A.~K. Tripathi, ``Is deep learning good enough for software defect prediction?'' \emph{Innovations in Systems and Software Engineering}, vol.~21, no.~2, pp. 501--516, 2025.

\bibitem{hu2025fed}
X.~Hu, M.~Zheng, R.~Zhu, X.~Zhang, and Z.~Jin, ``Fed-olf: federated oversampling learning framework for imbalanced software defect prediction under privacy protection,'' \emph{IEEE Transactions on Reliability}, 2025.

\bibitem{han2024bjcnet}
J.~Han, C.~Huang, and J.~Liu, ``bjcnet: A contrastive learning-based framework for software defect prediction,'' \emph{Computers \& Security}, vol. 145, p. 104024, 2024.

\bibitem{hesamolhokama2024sdperl}
M.~Hesamolhokama, A.~Shafiee, M.~Ahmaditeshnizi, M.~Fazli, and J.~Habibi, ``Sdperl: A framework for software defect prediction using ensemble feature extraction and reinforcement learning,'' \emph{arXiv preprint arXiv:2412.07927}, 2024.

\bibitem{malhotra2025dhg}
R.~Malhotra and P.~Singh, ``Dhg-bigru: Dual-attention based hierarchical gated bigru for software defect prediction,'' \emph{Information and Software Technology}, vol. 179, p. 107646, 2025.

\bibitem{nashaat2025refining}
M.~Nashaat and J.~Miller, ``Refining software defect prediction through attentive neural models for code understanding,'' \emph{Journal of Systems and Software}, vol. 220, p. 112266, 2025.

\bibitem{alon2019code2vec}
\BIBentryALTinterwordspacing
U.~Alon, S.~Brody, O.~Levy, and E.~Yahav, ``code2vec: Learning distributed representations of code,'' in \emph{Proceedings of the ACM on Programming Languages (POPL)}.\hskip 1em plus 0.5em minus 0.4em\relax ACM, 2019, pp. 1--29. [Online]. Available: \url{https://dl.acm.org/doi/10.1145/3290353}
\BIBentrySTDinterwordspacing

\bibitem{alon2018code2seq}
\BIBentryALTinterwordspacing
U.~Alon, M.~Zilberstein, O.~Levy, and E.~Yahav, ``code2seq: Generating sequences from structured representations of code,'' in \emph{International Conference on Learning Representations (ICLR)}, 2019. [Online]. Available: \url{https://openreview.net/forum?id=H1gKYo09tX}
\BIBentrySTDinterwordspacing

\bibitem{zhang2019astnn}
\BIBentryALTinterwordspacing
J.~Zhang, X.~Wang, H.~Zhang, H.~Sun, K.~Wang, and X.~Liu, ``A novel neural source code representation learning model for defect prediction,'' in \emph{Proceedings of the 41st International Conference on Software Engineering (ICSE)}.\hskip 1em plus 0.5em minus 0.4em\relax IEEE, 2019, pp. 906--917. [Online]. Available: \url{https://ieeexplore.ieee.org/document/8812061}
\BIBentrySTDinterwordspacing

\bibitem{feng2020codebert}
Z.~Feng, D.~Guo, D.~Tang, N.~Duan, X.~Feng, M.~Gong, L.~Shou, B.~Qin, T.~Liu, D.~Jiang \emph{et~al.}, ``Codebert: A pre-trained model for programming and natural languages,'' \emph{arXiv preprint arXiv:2002.08155}, 2020.

\bibitem{guo2021graphcodebert}
\BIBentryALTinterwordspacing
D.~Guo, S.~Ren, S.~Lu, D.~Tang, N.~Duan, M.~Gong, L.~Shou, D.~Jiang, and M.~Zhou, ``Graphcodebert: Pre-training code representations with data flow,'' in \emph{International Conference on Learning Representations (ICLR)}, 2021. [Online]. Available: \url{https://arxiv.org/abs/2009.08366}
\BIBentrySTDinterwordspacing

\bibitem{xian2024transformcode}
\BIBentryALTinterwordspacing
Z.~Xian, R.~Huang, D.~Towey, and C.~Fang, ``Transformcode: A contrastive learning framework for code embedding via subtree transformation,'' \emph{IEEE Transactions on Software Engineering}, 2024, in Press. [Online]. Available: \url{https://ieeexplore.ieee.org/abstract/document/10508627}
\BIBentrySTDinterwordspacing

\bibitem{xu2023xastnn}
\BIBentryALTinterwordspacing
Z.~Xu, M.~Zhou, X.~Zhao, Y.~Chen, and X.~Cheng, ``xastnn: Improved code representations for industrial practice,'' in \emph{Proceedings of the 31st ACM Joint European Software Engineering Conference and Symposium on the Foundations of Software Engineering (ESEC/FSE)}.\hskip 1em plus 0.5em minus 0.4em\relax ACM, 2023. [Online]. Available: \url{https://dl.acm.org/doi/10.1145/3611643.3613869}
\BIBentrySTDinterwordspacing

\bibitem{li2023contextuality}
\BIBentryALTinterwordspacing
Y.~Li, S.~Wang, and T.~N. Nguyen, ``Contextuality of code representation learning,'' in \emph{Proceedings of the 38th IEEE/ACM International Conference on Automated Software Engineering (ASE)}, 2023, pp. 1--12. [Online]. Available: \url{https://ieeexplore.ieee.org/abstract/document/10298296}
\BIBentrySTDinterwordspacing

\bibitem{wang2023tree}
\BIBentryALTinterwordspacing
W.~Wang, K.~Zhang, G.~Li, S.~Liu, and A.~Li, ``Learning program representations with a tree-structured transformer,'' \emph{IEEE Transactions on Software Analysis and Modeling}, 2023. [Online]. Available: \url{https://ieeexplore.ieee.org/document/10123596}
\BIBentrySTDinterwordspacing

\bibitem{ahmad2021unified}
W.~U. Ahmad, S.~Chakraborty, B.~Ray, and K.-W. Chang, ``Unified pre-training for program understanding and generation,'' \emph{arXiv preprint arXiv:2103.06333}, 2021.

\bibitem{chen2021evaluating}
M.~Chen, J.~Tworek, H.~Jun, Q.~Yuan, H.~P. de~Oliveira~Pinto, J.~Kaplan, H.~Edwards, Y.~Burda, N.~Joseph, G.~Brockman \emph{et~al.}, ``Evaluating large language models trained on code,'' in \emph{arXiv preprint arXiv:2107.03374}, 2021.

\bibitem{hu2023large}
X.~Hu, Z.~Li, G.~Li, Z.~Zhang, and Z.~Jin, ``Large language models for code: A survey,'' \emph{arXiv preprint arXiv:2307.07056}, 2023.

\bibitem{white2023chatgpt}
M.~White, M.~Tufano, B.~Chen, and D.~Poshyvanyk, ``Chatgpt and software engineering: Opportunities and challenges,'' \emph{IEEE Software}, 2023.

\bibitem{peng2023impact}
C.~Peng, J.~Li, X.~Xia, and D.~Lo, ``The impact of large language models on software engineering: Preliminary findings on github copilot,'' \emph{ACM Transactions on Software Engineering and Methodology (TOSEM)}, vol.~32, no.~3, pp. 1--29, 2023.

\bibitem{thota2020survey}
M.~K. Thota, F.~H. Shajin, P.~Rajesh \emph{et~al.}, ``Survey on software defect prediction techniques,'' \emph{International Journal of Applied Science and Engineering}, vol.~17, no.~4, pp. 331--344, 2020.

\bibitem{akimova2021survey}
E.~N. Akimova, A.~Y. Bersenev, A.~A. Deikov, K.~S. Kobylkin, A.~V. Konygin, I.~P. Mezentsev, and V.~E. Misilov, ``A survey on software defect prediction using deep learning,'' \emph{Mathematics}, vol.~9, no.~11, p. 1180, 2021.

\bibitem{kontarinis2014debate}
D.~Kontarinis, ``Debate in a multi-agent system: multiparty argumentation protocols,'' Ph.D. dissertation, Universit{\'e} Ren{\'e} Descartes-Paris V, 2014.

\bibitem{liang2024encouraging}
T.~Liang, Z.~He, W.~Jiao, X.~Wang, Y.~Wang, R.~Wang, Y.~Yang, S.~Shi, and Z.~Tu, ``Encouraging divergent thinking in large language models through multi-agent debate,'' in \emph{Proceedings of the 2024 conference on empirical methods in natural language processing}, 2024, pp. 17\,889--17\,904.

\bibitem{liu2024groupdebate}
T.~Liu, X.~Wang, W.~Huang, W.~Xu, Y.~Zeng, L.~Jiang, H.~Yang, and J.~Li, ``Groupdebate: Enhancing the efficiency of multi-agent debate using group discussion,'' \emph{arXiv preprint arXiv:2409.14051}, 2024.

\bibitem{breiman2001random}
L.~Breiman, ``Random forests,'' \emph{Machine learning}, vol.~45, no.~1, pp. 5--32, 2001.

\bibitem{hosmer2013applied}
D.~W. Hosmer~Jr, S.~Lemeshow, and R.~X. Sturdivant, \emph{Applied logistic regression}.\hskip 1em plus 0.5em minus 0.4em\relax John Wiley \& Sons, 2013.

\bibitem{cortes1995support}
C.~Cortes and V.~Vapnik, ``Support-vector networks,'' \emph{Machine learning}, vol.~20, no.~3, pp. 273--297, 1995.

\bibitem{openai2024textembedding3}
OpenAI, ``Text embeddings 3 models,'' \url{https://platform.openai.com/docs/models/text-embedding-3-small}, 2024, accessed: 2025-12-14.

\bibitem{reimers2019sentence}
N.~Reimers and I.~Gurevych, ``Sentence-bert: Sentence embeddings using siamese bert-networks,'' \emph{arXiv preprint arXiv:1908.10084}, 2019.

\bibitem{menzies2007data}
T.~Menzies, J.~Greenwald, and A.~Frank, ``Data mining static code attributes to learn defect predictors,'' \emph{IEEE Transactions on Software Engineering}, vol.~33, no.~1, pp. 2--13, 2007.

\bibitem{rahman2013how}
F.~Rahman and P.~Devanbu, ``How, and why, process metrics are better,'' in \emph{Proceedings of the 2013 International Conference on Software Engineering}.\hskip 1em plus 0.5em minus 0.4em\relax IEEE Press, 2013, pp. 432--441.

\bibitem{yan2020software}
H.~Yan, Y.~Li, M.~Zhang, and Q.~Wu, ``Software defect prediction using deep belief networks,'' \emph{Information and Software Technology}, vol. 123, p. 106312, 2020.

\bibitem{brown2020language}
T.~B. Brown, B.~Mann, N.~Ryder, M.~Subbiah, J.~Kaplan, P.~Dhariwal, A.~Neelakantan, P.~Shyam, G.~Sastry, A.~Askell \emph{et~al.}, ``Language models are few-shot learners,'' in \emph{Advances in Neural Information Processing Systems (NeurIPS)}, vol.~33, 2020, pp. 1877--1901.

\bibitem{li2024software}
Z.~Li, J.~Niu, and X.-Y. Jing, ``Software defect prediction: future directions and challenges,'' \emph{Automated Software Engineering}, vol.~31, no.~1, p.~19, 2024.

\bibitem{white2023prompt}
J.~White, ``A prompt pattern catalog to enhance prompt engineering with chatgpt,'' \emph{arXiv preprint arXiv:2302.11382}, 2023.

\bibitem{liu2023pre}
P.~Liu \emph{et~al.}, ``Pre-train, prompt, and predict: A survey of prompting methods in natural language processing,'' \emph{ACM Computing Surveys}, 2023.

\bibitem{wan2023efficientllm}
\BIBentryALTinterwordspacing
Z.~Wan, X.~Wang, C.~Liu, S.~Alam, Y.~Zheng, and J.~Liu, ``Efficient large language models: A survey,'' \emph{arXiv preprint arXiv:2312.03863}, 2023. [Online]. Available: \url{https://arxiv.org/abs/2312.03863}
\BIBentrySTDinterwordspacing

\bibitem{deepseek2025r1}
\BIBentryALTinterwordspacing
{DeepSeek AI}, ``Deepseek-r1: Reasoning language model,'' Technical report and model release, 2025. [Online]. Available: \url{https://github.com/deepseek-ai/DeepSeek-R1}
\BIBentrySTDinterwordspacing

\bibitem{deepseek2024v3}
\BIBentryALTinterwordspacing
------, ``Deepseek-v3 technical report,'' Technical report and model release, 2024. [Online]. Available: \url{https://github.com/deepseek-ai/DeepSeek-V3}
\BIBentrySTDinterwordspacing

\bibitem{deepseek2025v31terminus}
\BIBentryALTinterwordspacing
------, ``Deepseek-v3.1 671b terminus,'' Hugging Face model card, 2025. [Online]. Available: \url{https://huggingface.co/deepseek-ai}
\BIBentrySTDinterwordspacing

\bibitem{google2025gemini25flashlite}
\BIBentryALTinterwordspacing
{Google DeepMind}, ``Gemini 2.5 flash-lite,'' Google DeepMind Generative AI Models Documentation, 2025, accessed: 2025-12-19. [Online]. Available: \url{https://deepmind.google/models/gemini/flash-lite}
\BIBentrySTDinterwordspacing

\bibitem{mistralai2025mistralsmallinstruct2503}
\BIBentryALTinterwordspacing
{Mistral AI}, ``Mistral small 3.1 24b instruct,'' Instruction-tuned model release, 2025, instruction-tuned variant of Mistral Small 3.1. [Online]. Available: \url{https://huggingface.co/mistralai/Mistral-Small-3.1-24B-Instruct-2503}
\BIBentrySTDinterwordspacing

\bibitem{jiang2024mixtral}
\BIBentryALTinterwordspacing
A.~Q. Jiang, A.~Sablayrolles, A.~Mensch \emph{et~al.}, ``Mixtral of experts,'' arXiv preprint arXiv:2401.04088, 2024. [Online]. Available: \url{https://arxiv.org/abs/2401.04088}
\BIBentrySTDinterwordspacing

\bibitem{mistralai2024codestral}
\BIBentryALTinterwordspacing
{Mistral AI Team}, ``Codestral,'' Mistral AI blog / model announcement, May 2024, official model card / downloads available on Hugging Face: https://huggingface.co/mistralai/Codestral-22B-v0.1. [Online]. Available: \url{https://mistral.ai/news/codestral}
\BIBentrySTDinterwordspacing

\bibitem{qwen2025coder32binstruct}
\BIBentryALTinterwordspacing
{Alibaba / Qwen Team}, ``Qwen2.5-coder-32b-instruct,'' Hugging Face model card, 2025, accessed: 2025-12-19. [Online]. Available: \url{https://huggingface.co/Qwen/Qwen2.5-Coder-32B-Instruct}
\BIBentrySTDinterwordspacing

\bibitem{christen2023review}
P.~Christen, D.~J. Hand, and N.~Kirielle, ``A review of the f-measure: its history, properties, criticism, and alternatives,'' \emph{ACM Computing Surveys}, vol.~56, no.~3, pp. 1--24, 2023.

\bibitem{sokolova2009systematic}
M.~Sokolova and G.~Lapalme, ``A systematic analysis of performance measures for classification tasks,'' \emph{Information processing \& management}, vol.~45, no.~4, pp. 427--437, 2009.

\bibitem{tharwat2021classification}
A.~Tharwat, ``Classification assessment methods,'' \emph{Applied computing and informatics}, vol.~17, no.~1, pp. 168--192, 2021.

\bibitem{yang2025dlap}
Y.~Yang, X.~Zhou, R.~Mao, J.~Xu, L.~Yang, Y.~Zhang, H.~Shen, and H.~Zhang, ``Dlap: A deep learning augmented large language model prompting framework for software vulnerability detection,'' \emph{Journal of Systems and Software}, vol. 219, p. 112234, 2025.

\bibitem{Wilcoxonsigned}
\BIBentryALTinterwordspacing
F.~Wilcoxon, ``Individual comparisons by ranking methods,'' \emph{Biometrics Bulletin}, vol.~1, no.~6, pp. 80--83, 1945. [Online]. Available: \url{http://www.jstor.org/stable/3001968}
\BIBentrySTDinterwordspacing

\end{thebibliography}

\pagebreak
\appendix

\subsection{Unified Diff Example}
\label{appendix:unified_diff_exp}
See \autoref{fig:unified_diff_example}

\begin{figure}[h]
    \centering
    \begin{adjustbox}{width=\columnwidth}
        \begin{lstlisting}[style=diff,basicstyle=\ttfamily\scriptsize]
--- a/DefectExample.java
+++ b/DefectExample.java
@@ -3,10 +3,10 @@ public class DefectExample {
     public static int calculateSum(int[] arr) {
         int total = 0;
-        // Defective loop condition
-        for (int i = 0; i <= arr.length; i++) {
+        // Corrected loop condition
+        for (int i = 0; i < arr.length; i++) {
             total += arr[i];
         }
         return total;
     }
 }
        \end{lstlisting}
    \end{adjustbox}
    \caption{Example of a unified diff patch highlighting a single-line bug fix}
    \label{fig:unified_diff_example}
\end{figure}

\subsection{Parameters and Configuration of Models and Algorithms}
\label{appendix:parameters_config}

See \autoref{tab:experimental_settings} for parameters and configurations.

\textbf{Resources:} Experiments were conducted using Kaggle Notebooks with free cloud compute. The environment provided up to 12 hours per session, 20 GB disk space, approximately 30 GB RAM, and access to two NVIDIA Tesla T4 GPUs. The GPUs were used exclusively for CodeBERT embedding computation; all other components ran on CPU.

\begin{table}[h]
\centering
\caption{Hierarchical overview of versions and parameters used in the experimental setting. All parameters not explicitly listed are left at their default values.}
\label{tab:experimental_settings}
\footnotesize
\setlength{\tabcolsep}{3pt}
\renewcommand{\arraystretch}{1.05}
\begin{tabularx}{0.40\textwidth}{lX}
\toprule
\textbf{Component / Parameter} & \textbf{Value} \\
\midrule
\textit{Execution environment} & \\

\quad Python version & 3.11.13 (GCC 11.4.0) \\
\quad Notebook version & 6.5.4 \\
\quad JupyterLab version & 3.6.8 \\

\textit{Classification models} & \\

\quad \textit{Logistic Regression} & \\
\qquad Library & \textsc{scikit-learn} \\
\qquad Solver & "liblinear" \\
\qquad Max iterations & 2000 \\

\quad \textit{Lasso} & \\
\qquad Penalty & "l1" \\
\qquad Solver & "liblinear" \\
\qquad Max iterations & 2000 \\
\qquad Class weighting & "balanced" \\

\quad \textit{Random Forest} & \\
\qquad Number of estimators & 200 \\
\qquad Random seed & 42 \\
\qquad Class weighting & "balanced" \\

\quad \textit{Linear SVM} & \\
\qquad Kernel & "linear" \\
\qquad Probability estimates & True \\
\qquad Class weighting & "balanced" \\

\midrule
\textit{Code representations} & \\

\quad \textit{CodeBERT} & \\
\qquad Pretrained model & "microsoft/codebert-base" \\
\qquad Tokenizer & RobertaTokenizer \\
\qquad Encoder & RobertaModel \\
\qquad Token truncation & True \\
\qquad Maximum seq length & 512 \\
\qquad Padding strategy & "max\_length" \\
\qquad Tensor format & "pt" \\
\qquad dimension & 768 \\

\quad \textit{Long-file handling} & \\
\qquad Chunking method & \textsc{tiktoken} \\

\quad \textit{API-based embeddings} & \\
\qquad Embedding model & OpenAI text-embedding-3-small \\
\qquad dimension & 1536 \\

\midrule
\textit{LLM component} & \\
\quad Library & \textsc{openai}, \textsc{requests} \\
\quad LLM provider & llm7\footnotemark \\
\quad Temperature & 0.0 \\
\quad Base URL & "https://api.llm7.io/v1"\\

\midrule
\textit{Context Expansion} & \\

\quad depth & 3 \\
\quad max\_lines & 400 (default), 0, 100, 200, 600 (final), 800, 1000 \\

\bottomrule
\end{tabularx}
\end{table}

\footnotetext{
\texttt{llm7} is a third-party LLM aggregation platform (\url{https://llm7.io/}) that provides access to multiple large language models through a unified API. The platform is open-source and publicly available at \url{https://github.com/chigwell/llm7.io}.}

\subsection{Proof of Probabilistic Correctness of File Matching}
\label{appendix:matching_proof}

Fix a new-version file $x_2$ and let
$D_{\mathrm{old}}=\{z_1,\dots,z_m\}$ denote the set of old-version files.
For each $z\in D_{\mathrm{old}}$, define the Dice similarity
\[
S(z)=\frac{2|L(x_2)\cap L(z)|}{|L(x_2)|+|L(z)|}\in[0,1].
\]
Let $s_1\ge s_2\ge\cdots\ge s_m$ denote the ordered similarities, and define
\[
g_i = s_i - s_{i+1}, \qquad i=1,\dots,m-1.
\]

\paragraph{Concentration of similarity scores.}
For a fixed $z\in D_{\mathrm{old}}$, write
\[
|L(x_2)\cap L(z)|
=
\sum_{\ell\in L(x_2)} \mathbf{1}[\ell\in L(z)],
\]
where the indicator variables are independent Bernoulli random variables.
Hence $S(z)$ is a normalized sum of independent Bernoulli variables.

By Hoeffding’s inequality, there exists a constant $\kappa>0$ such that for all
$\epsilon>0$,
\begin{equation}
\Pr\!\left(|S(z)-\mathbb{E}[S(z)]|\ge\epsilon\right)
\;\le\;
2\exp(-\kappa \epsilon^2 n),
\label{eq:concentration}
\end{equation}
where $n=\min(|L(x_2)|,|L(z)|)$.

This exponential tail bound allows uniform control over all similarities using
union bounds.

\begin{theorem}
For any $\delta>0$, there exist constants $T\in(0,1)$ and $c\ge1$ such that
Algorithm~\ref{alg:file_matching} satisfies:
\begin{enumerate}
\item if a genuine predecessor exists, it is identified with probability at
      least $1-\delta$;
\item if no predecessor exists, the probability of declaring a match is at most
      $\delta$.
\end{enumerate}
\end{theorem}

\begin{proof}
Fix $\delta>0$. Choose $\epsilon>0$ such that
\begin{equation}
\mu^*-\mu_0 \ge 6\epsilon,
\label{eq:separation}
\end{equation}
and such that
\begin{equation}
2m\exp(-\kappa\epsilon^2 n) \le \frac{\delta}{2}.
\label{eq:union_choice}
\end{equation}

For instance, one can take
\[
\epsilon = \max\Bigg\{\frac{\mu^*-\mu_0}{6}, \;\; \sqrt{\frac{1}{\kappa n} \ln \frac{4 m}{\delta}} \Bigg\}.
\]
Notice that larger $m$ requires a larger $\epsilon$ to keep the union-bound probability small, while larger $n$ allows smaller $\epsilon$ because the concentration improves with more lines.

Define the threshold
\[
T := \mu^* - 3\epsilon.
\]

\paragraph{Case 1: A true predecessor exists.}
Let $x_1^*$ denote the true predecessor of $x_2$.

For each $z\in D_{\mathrm{old}}$, by \eqref{eq:concentration},
\[
\Pr\!\left(|S(z)-\mathbb{E}[S(z)]|\ge\epsilon\right)
\le
2\exp(-\kappa\epsilon^2 n).
\]
Applying the union bound over all $m$ files,
\begin{align*}
&\Pr\!\left(
\exists z\in D_{\mathrm{old}} :
|S(z)-\mathbb{E}[S(z)]|\ge\epsilon
\right) \\
&\le
\sum_{z\in D_{\mathrm{old}}}
\Pr\!\left(|S(z)-\mathbb{E}[S(z)]|\ge\epsilon\right) \\
&\le
2m\exp(-\kappa\epsilon^2 n)
\;\le\;
\frac{\delta}{2}.
\end{align*}
Therefore, with probability at least $1-\delta/2$,
\begin{equation}
S(x_1^*) \ge \mu^*-\epsilon,
\qquad
S(z) \le \mu_0+\epsilon
\quad \forall z\neq x_1^*.
\label{eq:good_event}
\end{equation}

On this event,
\[
s_1 = S(x_1^*), \qquad s_2 = \max_{z\neq x_1^*} S(z),
\]
and hence
\begin{equation}
s_1 \ge \mu^*-\epsilon \ge T,
\label{eq:s1_threshold}
\end{equation}
and
\begin{equation}
s_1-s_2
\ge (\mu^*-\epsilon)-(\mu_0+\epsilon)
= \mu^*-\mu_0-2\epsilon
\ge 4\epsilon,
\label{eq:top_gap}
\end{equation}
where the last inequality follows from \eqref{eq:separation}.

Now consider gaps $g_i$ for $i\ge2$.
For each such $i$, both $s_i$ and $s_{i+1}$ correspond to non-predecessor files.
Applying \eqref{eq:concentration} to all such similarities and using a union
bound over at most $m$ terms, we obtain that with probability at least
$1-\delta/2$,
\[
|s_i-\mathbb{E}[s_i]|\le\epsilon,
\qquad
|s_{i+1}-\mathbb{E}[s_{i+1}]|\le\epsilon
\quad \forall i\ge2.
\]
Consequently,
\begin{equation}
g_i \le 2\epsilon
\qquad \forall i\ge2.
\label{eq:small_gaps}
\end{equation}

Thus the empirical gap statistics satisfy
\[
\mu_g \le 2\epsilon,
\qquad
\sigma_g \le 2\epsilon.
\]
Choosing $c=1$, inequality \eqref{eq:top_gap} yields
\[
s_1-s_2 \ge \mu_g + c\sigma_g.
\]
Hence the algorithm returns the correct predecessor with probability at least
$1-\delta$.

\paragraph{Case 2: No predecessor exists.}
In this case, $\mathbb{E}[S(z)]\le\mu_0$ for all $z\in D_{\mathrm{old}}$.
Applying \eqref{eq:concentration} and a union bound over all $m$ files,
\begin{align*}
\Pr\!\left(s_1 > \mu_0+\epsilon\right)
&=
\Pr\!\left(\exists z\in D_{\mathrm{old}} :
S(z)-\mathbb{E}[S(z)]>\epsilon\right) \\
&\le
2m\exp(-\kappa\epsilon^2 n)
\;\le\;
\frac{\delta}{2}.
\end{align*}
Therefore, with probability at least $1-\delta/2$,
\begin{equation}
s_1 \le \mu_0+\epsilon < T,
\label{eq:null_case}
\end{equation}
and the algorithm returns \textbf{Null}.

A false match requires
\[
S(z)-\mathbb{E}[S(z)] \ge \mu^*-\mu_0-3\epsilon \ge 3\epsilon
\]
for some $z\in D_{\mathrm{old}}$.
By Hoeffding’s inequality and a union bound,
\[
\Pr(\text{false match})
\le
2m\exp(-\kappa\epsilon^2 n)
\le
\delta.
\]
\end{proof}

\subsection{Proof of \autoref{thm:label-persistence} (Bayes-Optimality of Naive Label-Persistent Classifier)}
\label{appendix:bayes_proof}

\begin{proof}
Fix $x$ such that $\mathbb P(X_{t+1}=x) > 0$.  
By the law of total probability over $Y_t \in \{0,1\}$, we can write
\begin{align*}
&\mathbb P(Y_{t+1}=1 \mid X_{t+1}=x) \\
&= \mathbb P(Y_{t+1}=1 \mid Y_t=0, X_{t+1}=x) \cdot \mathbb P(Y_t=0 \mid X_{t+1}=x) \\
&\quad + \mathbb P(Y_{t+1}=1 \mid Y_t=1, X_{t+1}=x) \cdot \mathbb P(Y_t=1 \mid X_{t+1}=x).
\end{align*}

Since $X_{t+1} = h(X_t,D)$ deterministically and $h$ is invertible, conditioning on $X_{t+1}=x$ restricts $(X_t,D)$ to the set
\[
\mathcal S_x := \{(u,d) : h(u,d) = x\}.
\]

Now consider the probability of label persistence for each $(X_t,D) \in \mathcal S_x$:

\[
\mathbb P(Y_{t+1}=Y_t \mid X_t=u, D=d)=
\begin{cases}
1, & d=\emptyset,\\
\le 1-\varepsilon, & d\neq\emptyset.
\end{cases}
\]

Let $p_0 = \mathbb P(D = \emptyset \mid X_{t+1}=x)$ denote the probability that the file did not change. Then, using the law of total probability over $D$ conditional on $Y_t$ and $X_{t+1}$, we have
\begin{align*}
&\mathbb P(Y_{t+1} = Y_t \mid Y_t, X_{t+1}=x) \\
&= \sum_{d \in \mathcal D} \mathbb P(Y_{t+1} = Y_t \mid Y_t, D=d, X_{t+1}=x) \\ &\cdot\mathbb P(D=d \mid Y_t, X_{t+1}=x) \\
&\ge p_0 \cdot 1 + (1-p_0) \cdot (1-\varepsilon) \\
&= 1 - (1-p_0)\varepsilon \ge 1-\varepsilon,
\end{align*}
because $0 \le p_0 \le 1$.

Finally, averaging over $Y_t$ (via the first law of total probability step) gives
\[
\mathbb P(Y_{t+1} = Y_t \mid X_{t+1}=x) \ge 1-\varepsilon > \frac12.
\]

Therefore, the Bayes-optimal classifier chooses the previous label:
\[
f^\star(X_{t+1}) = Y_t \quad \text{with probability at least } 1-\varepsilon.
\]
\end{proof}
\subsection{Proof of \autoref{thm:necessity-of-change} (Necessity of Conditioning on Change)}
\label{appendix:necessity_change}
\begin{proof}
Assume there exist distinct pairs $(X_t,D) \neq (X_t',D')$ such that
\[
h(X_t,D) = h(X_t',D') = X_{t+1}.
\]
Then the mapping from $(X_t,D)$ to $X_{t+1}$ is not injective in $D$. Hence, for a fixed value of $X_{t+1}$, multiple values of $D$ are possible with positive probability. This implies that $D$ cannot be expressed as a measurable function of $X_{t+1}$ alone.

Because defect transitions depend on $D$, the conditional distribution
$\mathbb P(Y_{t+1} \mid Y_t, D)$ varies across different values of $D$ that map to the same $X_{t+1}$. Conditioning only on $X_{t+1}$ therefore averages over
these distinct change events:
\[
\mathbb P(Y_{t+1} \mid X_{t+1})
=
\mathbb E\!\left[
\mathbb P(Y_{t+1} \mid Y_t, D)
\,\middle|\,
X_{t+1}
\right],
\]
which is generally different from $\mathbb P(Y_{t+1} \mid Y_t, D)$ for any
specific change $D$.
\end{proof}

\subsection{System Prompt for LLM Baselines}
\label{appendix:system_prompt}

All LLM-based baselines use the system prompt shown in \autoref{fig:system_prompt} to ensure consistent behavior across model.

\begin{figure}[ht]
\centering
\begin{promptbox}[System Prompt]
You are an expert software engineer and code reviewer.
Your task is to analyze source code and code changes to determine whether a file is Defective or Benign.
Carefully reason about correctness, logic, and potential defects based only on the provided information.
Do not assume missing context or speculate beyond the given input.
Return only the final classification when asked.
\end{promptbox}
\caption{System prompt used for all LLM-based baselines.}
\label{fig:system_prompt}
\end{figure}

\subsection{Prompts Used For LLM Baselines}
\label{appendix:prompts_llm}

\begin{figure}[ht]
\centering
\begin{promptbox}[Diff-guided reasoning]
You are given SRC1 (known to be Defective) and SRC2, along with the exact differences (added/removed/changed lines). Use these to decide the status of SRC2.

[SRC1] → [Defective]
[SRC2] → [???]

[SRC1]
\textless earlier version\textgreater

[SRC2]
\textless modified version\textgreater

[Differences]
\textless added/removed lines\textgreater

Review the differences and reason carefully. Is SRC2 Defective or Benign?
\end{promptbox}
\caption{Prompt used for the Diff-guided reasoning setting.}
\label{fig:prompt_m2}
\end{figure}

\begin{figure}[ht]
\centering
\begin{promptbox}[Direct comparison]
You are given SRC1 (known to be Defective) and SRC2 code. Compare the two versions and decide if SRC2 is Defective or Benign.

[SRC1] → [Defective]
[SRC2] → [???]

[SRC1]
\textless earlier version\textgreater

[SRC2]
\textless modified version\textgreater

Think step by step and decide whether SRC2 is Defective or Benign.
\end{promptbox}
\caption{Prompt used for the Direct comparison setting.}
\label{fig:prompt_m1}
\end{figure}

\begin{figure}[ht]
\centering
\begin{promptbox}[SRC1-only change reasoning]
You are only given SRC1 (known to be Defective) and the differences/unified diff. Based on how the code has changed, predict if the new SRC2 is Defective or Benign, even though SRC2 code is hidden.

[SRC1] → [Defective]
[SRC2] → [???]

[SRC1]
\textless earlier version\textgreater

[Differences]
\textless added/removed lines\textgreater

[Unified diff]
\textless patch-style diff\textgreater

Based only on the observed changes, predict whether the unseen SRC2 would be Defective or Benign.
\end{promptbox}
\caption{Prompt used for the SRC1-only change reasoning setting.}
\label{fig:prompt_m4}
\end{figure}

\begin{figure}[ht]
\centering
\begin{promptbox}[Diffs with exemplars]
You are given the differences and unified diff. SRC1 was Defective. Judge whether these changes are likely to change SRC2's status or keep it Defective. Some example defective code lines are also provided.

[SRC1] → [Defective]
[SRC2] → [???]

[Differences]
\textless added/removed lines\textgreater

[Unified diff]
\textless patch-style diff\textgreater

[Defective Examples]
\textless example buggy lines\textgreater

Compare the diffs and examples of defective code. Does SRC2 appear Defective or Benign?
\end{promptbox}
\caption{Prompt used for the Diffs with exemplars setting.}
\label{fig:prompt_m7}
\end{figure}

\begin{figure}[ht]
\centering
\begin{promptbox}[Local-context reasoning]
You are given only the locally relevant code changes with a few lines of context around them. SRC1 is known to be Defective. Based only on these local modifications, predict whether the updated SRC2 is Defective or Benign.

[SRC1] → [Defective]
[SRC2] → [???]

[Local Context]
\textless modified lines with nearby context\textgreater

[Differences]
\textless added/removed lines\textgreater

Considering only the local code context and changes, predict if SRC2 is Defective or Benign.
\end{promptbox}
\caption{Prompt used for the Local-context reasoning setting.}
\label{fig:prompt_m6}
\end{figure}

\begin{figure}[ht]
\centering
\begin{promptbox}[Patch-aware reasoning]
You are given SRC1 (known to be Defective), SRC2, the differences, and a unified diff (like a patch). Use all this information to determine if SRC2 is Defective or Benign.

[SRC1] → [Defective]
[SRC2] → [???]

[SRC1]
\textless earlier version\textgreater

[SRC2]
\textless modified version\textgreater

[Differences]
\textless added/removed lines\textgreater

[Unified diff]
\textless patch-style diff\textgreater

Use all provided information to conclude whether SRC2 is Defective or Benign.
\end{promptbox}
\caption{Prompt used for the Patch-aware reasoning setting.}
\label{fig:prompt_m3}
\end{figure}

\begin{figure}[ht]
\centering
\begin{promptbox}[Semantic repair reasoning]
You are given the differences, unified diff, and status of SRC1 (known to be Defective). Analyze the changes and determine whether SRC2 is Defective or Benign:
- Do the modifications fix an existing defect?
- Do they introduce a new defect?
- Or leave the code unchanged?

[SRC1] → [Defective]
[SRC2] → [???]

[Differences]
\textless added/removed lines\textgreater

[Unified diff]
\textless patch-style diff\textgreater

Think step by step and decide the final status of SRC2.
\end{promptbox}
\caption{Prompt used for the Semantic repair reasoning setting.}
\label{fig:prompt_m8}
\end{figure}

\end{document}